\documentclass[preprint,
superscriptaddress,
 amsmath,amssymb, aps, prx
]{revtex4-2}

\usepackage{hyperref}
\usepackage{graphicx}
\usepackage{dcolumn}
\usepackage{bm}
\usepackage{amsmath, amssymb, amsthm}
\usepackage{stmaryrd} 
\usepackage{mathrsfs}
\usepackage{cleveref}
\usepackage{braket}
\usepackage{subfig}
\usepackage{xcolor}

\newtheorem{theorem}{Theorem}[section]
\newtheorem{lemma}[theorem]{Lemma}

\newtheorem{definition}[theorem]{Definition}
\newtheorem{proposition}[theorem]{Proposition}

\newtheorem{assumption}[theorem]{Assumption}
\newtheorem{remark}{Remark}

\def\RR{\mathbb{R}}

\def\CC{\mathbb{C}}

\def\x{U}

\def\target{U_f}
\def\z{M}

\def\tr{\operatorname{Tr}}
\def\iu{\operatorname{i}}
\def\RE{\operatorname{Re}}

\def\d{\mathrm{d}}

\newcommand{\hnorm}[1]{\vert #1 \vert}

\begin{document}
\title{Fast and Smooth Quantum Unitary Control of High-Dimensional Spin-Qudits via a Shooting Technique}
\date{\today}

\author{Paul-Louis Etienney}
\email{paul-louis.etienney@ipcms.unistra.fr}
\affiliation{Université de Strasbourg, CNRS, Institut de Physique et Chimie des Matériaux de Strasbourg, UMR 7504, 67000 Strasbourg, France}
\author{Denis Jankovi{\'{c}}}
\email{denis.jankovic@qns.science}
\affiliation{Center for Quantum Nanoscience, Institute for Basic Science, Seoul 03760, Republic of Korea and Ewha Womans University, Seoul 03760, Republic of Korea}
\author{Killian Lutz}
\email{killian.lutz@inria.fr}
\affiliation{Université de Strasbourg, CNRS, INRIA, Institut de Recherche Mathématique Avancée, UMR 7501, 67000 Strasbourg, France}
\author{Jean-Gabriel Hartmann}
\email{jean-gabriel.hartmann@ipcms.unistra.fr}
\affiliation{Université de Strasbourg, CNRS, Institut de Chimie, UMR 7177, F-67000 Strasbourg, France}
\author{Emmanuel Franck}
\affiliation{Université de Strasbourg, CNRS, INRIA, Institut de Recherche Mathématique Avancée, UMR 7501, 67000 Strasbourg, France}
\author{Paul-Antoine Hervieux}
\email{paul-antoine.hervieux@ipcms.unistra.fr}
\affiliation{Université de Strasbourg, CNRS, Institut de Physique et Chimie des Matériaux de Strasbourg, UMR 7504, 67000 Strasbourg, France}

\begin{abstract}
High-fidelity quantum control is a cornerstone of scalable quantum technologies. We introduce a shooting-based optimization framework that generates smooth, experimentally realistic control pulses for implementing quantum gates in discrete quantum systems. Through numerical simulations on realistic architectures inspired by single-molecule magnets, we demonstrate that our method efficiently decomposes target quantum operations into electric pulse sequences while outperforming the widely used GRAPE algorithm. 
By exploiting the geometry of the unitary group, our approach produces hardware specific control pulses that are both faster and smoother, with performance gains that become increasingly significant as the dimensionality of the quantum system grows.
\end{abstract}
\maketitle

\tableofcontents
\section{Introduction}

Quantum optimal control provides a powerful framework for steering quantum systems toward specified target states and implementing desired operations (quantum gates) with high fidelity, while minimizing the operation time, accounting for physical constraints or environmental noise~\cite{Quantum_optimal_control_review, Quantum_optimal_control_review2}. 
To implement an operation, the corresponding quantum gate has to be unravelled into a sequence of electromagnetic control pulses sent to the system.
Multiple algorithms have been designed to address this issue, most notably the widely used GRadient Ascent Pulse Engineering (GRAPE) \cite{GRAPE_original} and methods derived from the work of Krotov~\cite{Krotov, Krotov_python, Le2026}.
In this work, we present MAGICARP, a novel quantum optimal control algorithm combining natural gradient descent~\cite{nurbekyan2023efficient} and a shooting technique~\cite{shooting_approach, A_time-parallel_multiple-shooting_method}, and benchmark it against GRAPE.


We will be using molecular spin qudits based on single-molecule magnets (SMMs) as a representative quantum-computing platform. Indeed, nuclear spins of SMMs based on lanthanides (Ln) have experienced a surge of scientific interest and technical development in recent years~\cite{Moreno-Pineda2021}. These systems possess a number of favorable properties for quantum information processing (QIP): the nuclear spin states are highly stable, exhibiting long coherence times that are improved due to the environmental isolation afforded by the molecular ligand; the $4f$ valence electrons in the triply-ionized bound state provide useful means of electric control~\cite{Thiele2014}; and the high-dimensional Hilbert space (qudit) of the individually addressable hyperfine levels enables novel computing
applications~\cite{godfr2017}. In this study, we focus on terbium ions sandwiched between phthalocyanine (Pc) ligand layers, commonly referred to in the literature as “deckers”~\cite{SMMs_in_spintronics_plus_history}, which have been extensively investigated experimentally. Furthermore, working with multi-level quantum systems with dimensions larger than two has been shown to have some benefits over qubit systems, such as~\cite{Moreno-Pineda2020}: a reduced number of elementary gates to perform operations; lower error rates compared with qubits; more complex gates can be applied to a single qudit; entanglement and superposition can be achieved with smaller clusters of processing units; higher information density per processing unit; and, subsequently,  increased data rates of communication channels. Considering that these individual hyperfine states can be addressed and manipulated using microwave electric pulses, the latter can be shaped using optimal control techniques for quantum optimization and will enable faster logical gates. However, as the dimension of the quantum state space increases, one of the major challenges is the growing amount of time required to perform operations. It is therefore important to find algorithms that can perform this task as quickly as possible. Furthermore, from an experimental standpoint, signals with the least abrupt temporal variations should be preferred. Thus, developing algorithms that generate protocols with smooth control signals is also a key challenge. 



We benchmark our shooting-based approach against GRAPE on two physically relevant molecular systems, the double- and triple-decker compounds TbPc$_2$ and Tb$_2$Pc$_3$, characterized by Hilbert-space dimensions of 4 and 16, respectively. We chose to numerically compare the results of our method to the GRAPE algorithm, since both can be used without assuming specialized knowledge on the system and implement gradient-based optimisations at their cores. Our numerical results demonstrate that this shooting method achieves the desired fidelity threshold while providing faster execution times (gate times). The effects of a decoherence channel on the system are all the more mitigated by this figure of merit, especially in large qudit systems~\cite{Jankovic_Nature, Hartmann2025nonlinearityof}.

More generally, the ideas developed here could in principle be adapted to other quantum systems with an addressable finite-dimensional Hilbert space~\cite{Quantum_optimal_control_review2}, including nitrogen-vacancy centers in diamond~\cite{diamond_centers}, NMR systems~\cite{NMR}, and trapped ions~\cite{trapped_ions}.

This paper is organized as follows. In~\cref{sec:intro_and_pres}, we introduce three pulse implementation techniques used in the numerical comparison and explain how we compare them. More precisely, in~\cref{subsec:other_methods} we present two existing methods, the Givens rotation decomposition (GRD) and the GRAPE algorithm, while in~\cref{subsec:Magicarp_pres} we introduce our proposed method MAGICARP for the parametrization of controls.
We will compare these methods in~\cref{sec:Numerical_implementation_on_SMM} using the $\text{TbPc}_2$ and $\text{Tb}_2\text{Pc}_3$ molecules presented in~\cref{subsec:SMM} and discuss the numerical results in~\cref{subsec:Magicarp_results}.

\section{Methodology}~\label{sec:intro_and_pres}
\subsection{Reference methods for control pulse optimization}~\label{subsec:other_methods}
In several platforms, 
quantum gates are realized by control pulses which are generated by an arbitrary waveform generator (AWG). Such a device can accurately control several parameters for each transition $n$ and time step $i$:
a frequency $\omega$ and phase $\theta$, the amplitude of the pulse $A$ and the step duration $\tau$.
In this setting and denoting by $\Theta$ the Heaviside function, the pulse function $S(t)$ is represented by:
\begin{align}\label{eq:pulse}
    S(t) = \sum_{n, i}^{}\Theta(t-t_{n, i}) \Theta(t-t_{n, i} - \tau_{n, i}) A_{n, i} \cos(\omega_{n,i} t - \theta_{n, i}).
\end{align}
Consider a simple two-level system, with a drift Hamiltonian $H_0$ describing a precession around the $z$-axis at frequency $\omega_{\text{prec}}$. We represent the AWG effect by a monochromatic electromagnetic field $\overrightarrow{E}$ 
(in a direction orthogonal to the precession, the $x$-axis for instance) interacting with the electric dipole moment $\overrightarrow{d}$ of the system~\cite{godfr2017}. This interaction can be described by a specific phase $\theta$, a Rabi amplitude $\Omega_{\text{Rabi}}$ for a frequency $\omega$ using the following Hamiltonian:
\begin{align}\label{eq:Hamiltonian_btw_2_levels}
    H &= H_0 - \overrightarrow{d} \cdot \overrightarrow{E}\\
    &= \frac{\hbar\omega_{\text{prec}}}{2} \sigma_z + 2\hbar\Omega_{\text{Rabi}} \cos(\omega t - \theta) \sigma_x .\notag
\end{align}
By choosing a frame rotating at frequency $\omega_{\text{rot}}$, this Hamiltonian rewrites
\begin{align}\label{eq:labframeH}
    H' &= \hbar\frac{(\omega_{\text{prec}} - \omega_{\text{rot}})}{2}\sigma_z + \hbar\Omega_{\text{Rabi}} \left[(\cos((\omega + \omega_{\text{rot}})t - \theta) + \cos((\omega - \omega_{\text{rot}})t - \theta)) \sigma_x\right.\\
    &+ \left.(\sin((\omega + \omega_{\text{rot}})t - \theta) - \sin((\omega - \omega_{\text{rot}})t - \theta)) \sigma_y\right]. \notag
\end{align}
Usually three simplifications are made: (i) choosing $\omega_{\text{rot}}=\omega_{\text{prec}}$ to remove the drift term which places us in the interaction picture, (ii) choosing $\omega=\omega_{\text{rot}}$, which places us in the rotating frame, and (iii) applying the Rotating Wave Approximation (RWA) to neglect fast oscillating terms. The Hamiltonian then simplifies to
\begin{align}\label{eq:godfrin}
    H = \hbar\Omega_{\text{Rabi}} \left[\cos(\theta) \sigma_x - \sin(\theta) \sigma_y\right].
\end{align}
In particular $\omega=\omega_{\text{prec}}$ and so the pulse is on resonance with the transition. In the remainder of this paper, we place ourselves in a situation where we can apply~\cref{eq:godfrin} between any two energy levels of the Hilbert space we are working with as long as they are considered ``connected'', that is as long as such an interaction is experimentally allowed. We can control both $\Omega_{\text{Rabi}}$ and $\theta$ through an AWG.\ 

\subsubsection{Linear algebra: Givens Rotation Decomposition}\label{sec:GRD}
The Givens Rotation Decomposition~\cite{Givens1958} (GRD) is a method that eliminates subdiagonal elements of a matrix by sequentially applying rotation matrices operating on some two-dimensional subspace to decompose any matrix into a so-called QR decompostion~\cite{QR_decompostion} with $Q$ unitary and $R$ upper triangular. For a unitary matrix this amounts to decomposing the matrix into $Q$ unitary and $R$ diagonal. 
A Givens Rotation (GR), is defined by two real parameters $\theta$ and $\phi$ as
\begin{align}\label{eq:GRD}
    \text{GR}(\phi, \theta) = e^{-i(\cos(\theta)\sigma_x + \sin(\theta)\sigma_y)\phi}&=
    \begin{bmatrix}
    \cos(\phi) & -ie^{-i\theta}\sin(\phi) \\
    -ie^{i\theta}\sin(\phi) & \cos(\phi) 
    \end{bmatrix},
\end{align}
and is embedded in a larger space to zero out one element of a column (see the first step in~\cref{eq:GRD_zero_one_number}). Successive rotations then zero out all subdiagonal elements of a given column (see the second step in~\cref{eq:GRD_second_step}). We can proceed similarly with the next column, without affecting the zeros on the previous columns since the GR does not modify them. Given a unitary matrix $U$, these steps can be schematically expressed as follows, where the generic coefficients $a_{i,j}$, $b_{i,j}$, $c_{i,j}$ represent the modified matrix elements at different stages of the method:
\begin{align}\label{eq:GRD_zero_one_number}
\operatorname{GR}_1U &
=
\begin{pmatrix}
a_{1,1} & a_{1,2} & \cdots & a_{1,N} \\
a_{2,1} & a_{2,2} & \cdots & a_{2,N} \\
\vdots  & \vdots  & \ddots & \vdots \\
0 & a_{N,2} & \cdots & a_{N,N}
\end{pmatrix}
\\
\vdots \notag \\ 
\operatorname{GR}_{d-1} \ldots \operatorname{GR}_2 \operatorname{GR}_1 U& =
\begin{pmatrix}
b_{1,1} & b_{1,2} & \cdots & b_{1,N} \\
0 & b_{2,2} & \cdots & b_{2,N} \\
\vdots  & \vdots  & \ddots & \vdots \\
0 & b_{N,2} & \cdots & b_{N,N}
\end{pmatrix} \label{eq:GRD_second_step}
\\
\vdots \notag \\ 
\operatorname{GR}_{d(d-1)/2} \ldots \operatorname{GR}_2 \operatorname{GR}_1 U& =
\begin{pmatrix}
c_{1,1} & c_{1,2} & \cdots & c_{1,N} \\
0 & c_{2, 2} & \cdots & c_{2,N} \\
\vdots  & \vdots  & \ddots & \vdots \\
0 & 0 & \cdots & c_{N,N}
\end{pmatrix}= \begin{pmatrix}
e^{-i\xi_1} & 0 & \cdots & 0 \\
0 & e^{-i \xi_2} & \cdots & 0 \\
\vdots  & \vdots  & \ddots & \vdots \\
0 & 0 & \cdots & e^{-i\xi_N}
\end{pmatrix}.\label{eq:localphases}
\end{align}

In the end, all subdiagonal elements are zeroed out. Since the GRs are unitary and the product of unitary matrices remains unitary, the final matrix in~\cref{eq:localphases} is unitary. We deduce that coefficients above the diagonal are also null, leaving only diagonal elements of modulus one. 
We could make the $z$ rotations needed to remove the diagonal elements $e^{-i\xi_k}$ using some combinations of $x$ and $y$ rotations (e.g., via an Euler-angle decomposition) and would end up with the identity matrix in the right-hand side of~\cref{eq:localphases}. This means that the applied operations have constructed $U^\dagger$, up to a global phase. Consequently, starting from the conjugate transpose of a target unitary matrix $U$ in the GRD, we can construct $U$.

An important observation in~\cref{eq:GRD} is that, when the Hamiltonian in~\cref{eq:godfrin} is used to generate the GR matrix, the pulse amplitude $\Omega$ and duration $t$ are not independent; they enter only through their product, which defines the angle $\phi=\Omega t$. Put differently, a longer pulse of smaller amplitude is equivalent to a shorter pulse of larger amplitude as long as their products are equal. This would not be the case if we were to consider a drift term. However, in practice the power output of the experimental devices must not exceed a certain threshold, otherwise the pulses could damage the molecule or induce undesired transitions between energy levels~\cite{Optimal_control_ex_luis}. Moreover, mathematically, a higher amplitude pulse would induce a non-negligible Bloch-Siegert effect~\cite{Bloch1940} due to the fast-oscillating terms neglected in the RWA, which would make~\cref{eq:godfrin} not valid anymore. It is therefore not possible to generate quantum gates arbitrarily fast simply by using pulses of larger amplitudes.

The GRD algorithm relies only on linear algebra, is interpretable, and terminates in a finite number of steps but these steps must be applied sequentially. On the other hand, the techniques relying on optimization generally do not terminate in a finite number of steps and could be attracted to local extrema, but they are able to handle multi-chromatic pulses and thereby achieve significantly faster implementations of gates.

\subsubsection{Optimisation: GRadient Ascent Pulse Engineering (GRAPE)}\label{subsec:GRAPE}
The GRAPE method decomposes a quantum operation into a sequence of multi-chromatic pulses using gradient descent~\cite{GRAPE_original}. For the numerical comparison in~\cref{subsec:Magicarp_results} we will use the GRAPE method from the QuTiP library, since this package is widely used~\cite{QuTiP}. Assume that one can model a physical system by a Hamiltonian composed of $N_{\text{ctrls}}$ available control Hamiltonians $H_k$ controlled by some continuous and dimensionless functions $u_k(t)$ with an amplitude $\Omega$. The objective is to reach a target gate $U_f$
by constructing a gate $U$ which is as close as possible to $U_f$ and defined 
by \footnote{Here, $\mathcal{T}$ denotes the time-ordering operator, which rearranges products so that later times appear to the left; for instance, for small $\delta t$ one has $\mathcal{T}\left\{e^{-\frac{i}{\hbar}H(t_a)\delta t}\,e^{-\frac{i}{\hbar}H(t_b)\delta t}\right\}=e^{-\frac{i}{\hbar}H(\max\{t_a,t_b\})\delta t}\,e^{-\frac{i}{\hbar}H(\min\{t_a,t_b\})\delta t}$. This is needed whenever $[H(t),H(t')]\neq 0$, in which case the exponential in~\cref{eq:Hkuk} is understood as a (Dyson) time-ordered exponential.}:
\begin{align}\label{eq:Hkuk}
    H(t) &=\hbar\Omega\sum_{k=1}^{N_{\text{ctrls}}}u_k(t)H_k \;,\\
    U &= \mathcal{T}e^{-\frac{i}{\hbar}\int_{t=0}^{T}H(t)dt}\notag \;.
\end{align} 
We assume that there is no drift Hamiltonian by writing~\cref{eq:Hkuk} without such a term. Since determining optimal controls in closed-form is usually out of reach, we choose a discretization scheme and numerically optimize discrete control functions. In the GRAPE method, control functions are replaced by piecewise constant functions: for each piece, the duration is fixed to $\Delta t$ but the amplitudes are optimization variables. The target gate $U_f$ has to be reached by $U$ which is the combination of $N$ gates $U_j$ which are given by
\begin{align}\label{eq:Hj}
    H(t_j\leq t<t_{j+1})&=\hbar\Omega \sum_{k=1}^{N_{\text{ctrls}}}u_{j,k}H_k,\\
    U_j&=e^{-\frac{i}{\hbar}\Delta t H(t_j\leq t<t_{j+1})}\notag ,\\
    U &= U_N U_{N-1} \hdots U_2 U_1.\notag 
\end{align} 
In a Hilbert space of dimension $d$, the cost function to be minimized reflects the error between $U$ and the target gate, up to a global phase. With $\mathbf{u} = \{u_{j,k}\}_{j=1,\dots,N}^{k=1,\dots,N_{\text{ctrls}}}$, a common cost function is
\begin{equation}
	\text{Cost}(\mathbf{u}) = 1 - \frac{1}{d}\left|\operatorname{Tr}(U_f^\dagger U)\right|\;,
\end{equation}
and referred to as infidelity, following~\cite{QuTiP}.
To update the control terms $u_{j,k}$ different optimization algorithms could be used. For example, QuTiP~\cite{QuTiP} uses the SciPy~\cite{SciPy} implementation of the \href{https://qutip.org/docs/4.0.2/guide/guide-control.html#the-grape-algorithm}{L-BFGS-B}~\cite{L-BFGS-B-Scipy} method.

\subsubsection{Comparing execution times}~\label{sec:GRDvsGRAPE}
To ensure a fair comparison between the execution times obtained with the
different methods considered in this work, we impose the same amplitude
constraint on all control protocols. This constraint is motivated by the
analytical GRD construction, in which each elementary control pulse is
monochromatic, has fixed physical amplitude $\Omega$, and has a
duration chosen according to the desired rotation angle. In particular,
once the angles $\theta$ and $\phi$ have been determined from~\cref{eq:GRD},
the pulse duration is chosen such that $\phi = \Omega t$ where $[\Omega] = T^{-1}$.

By contrast, optimal-control methods such as GRAPE produce discretized
multichromatic pulses. During each time step $\Delta t$, several control
components can be applied simultaneously, with amplitudes that are varied
independently by the optimization algorithm. In order to compare such
pulses with the GRD protocol on equal footing, each time step must
therefore be rewritten as an equivalent pulse satisfying the same total
amplitude bound.

Let $\Omega_0$ be a total
amplitude fixed to 1, with $[\Omega_0] = T^{-1}$. We rewrite the Hamiltonian in~\cref{eq:Hj} during a given time step $j$ as
\begin{align*}
    H(t_j\leq t<t_{j+1})&=\hbar\Omega_0\sum_{k=1}^{N_{\text{ctrls}}} u_{j,k} H_k=\hbar\Omega_0
    \sum_{l=1}^{N_{\mathrm{trans}}}
    \left(
        x_{j,l} \sigma_{x,l} + y_{j,l} \sigma_{y,l}
    \right),
\end{align*}
The $u_{j,k}$ are the discretized amplitude found by the optimization process, that we rewrite using $N_{\mathrm{trans}}$, the number of transition frequencies allowed by the experimental system. In the case of~\cref{eq:godfrin} each transition frequency allows for 2 dimensionless control amplitudes, $x_{j,l}$ and $y_{j,l}$, so $N_{\mathrm{trans}}=N_{\mathrm{ctrls}}/2$.
In the following we rename $H(t_j\leq t<t_{j+1})$ as $H$ and get rid of the subscript $j$ to simplify the notations, though we are still working on a given time step. For each frequency component $l$, we define the dimensionless amplitude
\begin{equation}
    \mathfrak u_l=\sqrt{x_{l}^2 + y_{l}^2},
\end{equation}
and the phase angle $\theta_l$ given by
\begin{equation}
    \cos \theta_l = \frac{x_{l}}{\mathfrak u_l},
    \qquad
    \sin \theta_l = \frac{y_{l}}{\mathfrak u_l}.
\end{equation}
The
Hamiltonian can then be written as
\begin{align}
    H
    &=
    \hbar \Omega_0
    \sum_{l=1}^{N_{\mathrm{trans}}}
    \mathfrak u_l
    \left(
        \cos (\theta_l) \, \sigma_{x,l}
        +
        \sin (\theta_l) \, \sigma_{y,l}
    \right).
\end{align}

We now introduce the total dimensionless amplitude of the multichromatic
pulse,
\begin{equation}
    \mathfrak u=\sqrt{\sum_{l=1}^{N_{\mathrm{trans}}}\mathfrak u_l^2}.
\end{equation}
This allows us to separate the total amplitude from the relative weights of
the different frequency components:
\begin{align}
    H
    &=
    \hbar\Omega_0 \mathfrak u
    \sum_{l=1}^{N_{\mathrm{trans}}}
    \frac{\mathfrak u_l}{\mathfrak u}
    \left(
        \cos (\theta_l) \, \sigma_{x,l}
        +
        \sin (\theta_l) \, \sigma_{y,l}
    \right).
\end{align}
By construction, the normalized coefficients satisfy
\begin{equation}
    \sum_{i=1}^{N_{\mathrm{trans}}}
    \left(
        \frac{\mathfrak u_i}{\mathfrak u}
    \right)^2
    =
    1 \;,
\end{equation}
and if we define $E$ as the energy of the pulse over the timestep $\Delta t$ one has
\begin{equation}
    E = \hbar\Omega_0 \sqrt{\sum_{l=1}^{N_\mathrm{trans}} \mathfrak{u}_l^2} = \hbar\Omega_0 \mathfrak{u} \;.
\end{equation}
Let $E_\text{max}$ be the maximum acceptable energy for the system, imposing $\Omega = E_\text{max}/\hbar$ as the maximum physical amplitude allowed by the experimental setup, one can define the corresponding dimensionless maximum amplitude as
\begin{equation}
    \widetilde{\Omega}=\frac{\Omega}{\Omega_0}.
\end{equation}
The original discretized time step produces the unitary
\begin{equation}
    U(\Delta t)
    =
    \exp\left[
        -i
        \Omega_0 \mathfrak u
        \left(
        \sum_{i=l}^{N_{\mathrm{trans}}}
        \frac{\mathfrak u_l}{\mathfrak u}
        \left(
            \cos (\theta_l) \, \sigma_{x,l}
            +
            \sin (\theta_l) \, \sigma_{y,l}
        \right)
        \right)
        \Delta t
    \right].
\end{equation}
The same unitary can be generated with a pulse whose total physical
amplitude is bounded by $\Omega$, provided that the duration is
rescaled to
\begin{equation}
    \Delta t'
    =
    \frac{\mathfrak u}{\widetilde{\Omega}}
    \Delta t
    =
    \frac{\Omega_0 \mathfrak u}{\Omega}
    \Delta t .
\end{equation}
Indeed,
\begin{align}
    U(\Delta t)
    &=
    \exp\left[
        -i
        \Omega
        \left(
        \sum_{l=1}^{N_{\mathrm{trans}}}
        \frac{\mathfrak u_l}{\mathfrak u}
        \left(
            \cos (\theta_l) \, \sigma_{x,l}
            +
            \sin (\theta_l) \, \sigma_{y,l}
        \right)
        \right)
        \Delta t'
    \right].
\end{align}

We can verify that the energy of the new rescaled pulse is therefore the maximal allowed energy $E_\text{max}$ leading to
\begin{equation}
    \hbar\Omega \sqrt{\sum_{l=1}^{N_{\mathrm{trans}}}
    \left(
        \frac{\mathfrak u_l}{\mathfrak u}
    \right)^2} = \hbar\Omega = E_\text{max} \;.
\end{equation}
This rescaling expresses the equivalence between using a larger control
amplitude for a shorter time and using the maximum allowed amplitude for a
longer time and ensures a fair comparison between algorithms, but is relevant only when the drift Hamiltonian can be neglected. We chose this framework for the simplicity it provides, but the optimal method we present in the following can be adapted to the case where a drift Hamiltonian cannot be neglected.

\subsection{A novel optimization method based on a shooting method}~\label{subsec:Magicarp_pres}
The Method for Adjoint and Gradient-based self-Iterative Construction And Refinement of Pulses (MAGICARP) is a method described in~\cite{Magicarp} and named as such to distinguish it from conventional gradient techniques. The core of this article lies in its numerical implementation for a system inspired by the physics of SMMs, especially for the triple and double decker systems detailed in~\cref{subsec:SMM}.

\subsubsection{Heuristic explanation of the Pontryagin Maximum Principle}
The MAGICARP consists of a finite dimensional parametrization of controls in feedback form which is derived from the Pontryagin Maximum Principle (PMP), a mathematical tool used to reframe an optimisation problem on trajectories (living in an infinite dimensional space) into a pseudo-Hamiltonian problem where only a finite amount of coefficients, encoding the trajectories, are to be optimized. This change of point of view is analogous to using the principle of least action in classical physics. The optimization of the control fields can be recast as the following problem, with $F$ the function defining the dynamics of $U$:
\begin{align*}
    \forall t,\, \hbar \dot{U}(t) &= F(U(t), u(t), t),\\
    U(0) &= U_0,\\
    U(t_f) &= U_f,\\
    \forall t,\, u(t) &\in \mathbb{R}^{N_{\text{ctrls}}} \;,
\end{align*}
with $u(t)$ a vector of all the $u_k(t)$ for $k\leq N_{\text{ctrls}}$, $U(t)$ the gate obtained at time $t$ via the optimisation protocol, $U_0$ that is usually the identity and $U_f$ the target gate. 
To select a control among all those steering the system between the desired states, we seek to minimize a cost function $\mathcal{C}$ including a terminal cost $G$ and a running cost $F_0$ reflecting what a desirable control is:
\begin{equation}
\mathcal{C} = G(U(t_f), t_f) + \int_{0}^{t_f} F_0(U(t), u(t), t) \, dt.
\end{equation}
To achieve the transfer as quickly as possible, $G$ could either be the duration $t_f$ or, if we drop the hard constraint $U(t_f) = U_f$, the infidelity between $U(t_f)$ and $U_f$. For instance, in $F_0$, we could capture that the pulse energy has to be minimized by choosing $F_0(u(t))=\hbar \Omega_0 \sqrt{\sum_k {u_k^2(t)}}$.
If the pulse amplitude exceeds $\Omega$, the maximum allowed pulse amplitude, the controls can always be rescaled and applied over a longer duration, as in GRAPE, since amplitude and time are equivalent through the relation discussed in~\cref{sec:GRD} and detailed in~\cref{sec:GRDvsGRAPE}.

The PMP~\cite{pontryagin1962optimal} is an important tool of optimal control theory. To strengthen the physical intuition behind it, we provide a heuristic derivation inspired by~\cite{Sugny_PMP} of the corresponding necessary conditions of optimality by connecting it to the Euler-Lagrange equations from the calculus of variations.
Starting from static initial and final target quantum gates (that is, static boundary conditions), it converts the problem of optimizing a cost function into the problem of finding the extremum of an action.

\begin{definition}
    Let $U(t) \in \mathbb{R}^{n \times n}$ be the quantum gate we construct at time $t$, $\Lambda(t) \in \mathbb{R}^{n \times n}$ an adjoint state and $u(t) \subset \mathbb{R}^{N_{\rm ctrls}}$ the vector of controls $u_k(t)$ in \cref{eq:Hkuk}. Let $G$ be a terminal cost function, for instance the infidelity, $F_0$ a running cost and $F$ a vector function describing the system dynamics. An optimal control is defined as control that is an extremum of the action $S$ given by:
    \begin{align*}
        S &= \int_0^{t_f} dt \mathcal{L} + G(U(t_f), t_f),\\
        \mathcal{L} &=F_0(U(t), u(t), t) + \Lambda(t) \left( \hbar \dot{U}(t) - F(U(t), u(t), t) \right)\\
        \hbar \Lambda &= \frac{\partial \mathcal{L}}{\partial \dot{U}} \;.
    \end{align*}
\end{definition}
The function under the integral is the Lagrangian $\mathcal{L}$ of the optimal control problem. 
The first variation $\delta S$ of the action in response to a perturbation $\delta u$ of the control is formally given in terms of the state variation $\delta U$ and adjoint $\Lambda$ by
\begin{align*}
\delta S &= \int_{0}^{t_f} dt\left[ \frac{\partial F_0}{\partial U} \delta U + \frac{\partial F_0}{\partial u} \delta u + \delta \Lambda (\hbar \dot{U} - F) + \Lambda \left( \hbar \delta \dot{U} - \frac{\partial F}{\partial U} \delta U - \frac{\partial F}{\partial u} \delta u \right) \right] + \frac{\partial G}{\partial U(t_f)} \delta U(t_f),\\
 &= \int_{0}^{t_f} dt \Bigg( \left[ \frac{\partial F_0}{\partial U} - \hbar\dot{\Lambda} - \Lambda \frac{\partial F}{\partial U} \right] \delta U 
+ \left[\hbar \dot{U} - F\right] \delta \Lambda 
+ \left[ \frac{\partial F_0}{\partial u} - \Lambda \frac{\partial F}{\partial u} \right] \delta u \Bigg)\\
\quad & + \left( \frac{\partial G}{\partial U(t_f)} + \hbar \Lambda(t_f) \right) \delta U(t_f) + \hbar \Lambda(0)\delta U(0),
\end{align*}
by integration by parts of the term $\Lambda \delta \dot{U}$. Moreover, $\delta U(0)=\delta U(t_f)=0$ since the propagator starts from the identity and is required to reach the target gate. Since $\delta S = 0$, the extrema of $S$ fulfill the following conditions.
\begin{proposition}
If a control $u$ associated to the state $U$ is optimal, then there exists an adjoint state $\Lambda$ such that the following conditions hold:\\
        \begin{align}
            \dot{\Lambda} &= \frac{\partial F_0}{\partial U} - \Lambda \frac{\partial F}{\partial U}, \label{eq:1} \\
            \hbar \dot{U} &= F, \label{eq:2} \\
            \frac{\partial F_0}{\partial u} &= \Lambda \frac{\partial F}{\partial u}, \label{eq:3}\\
            U(0) &= U_0, \quad U(t_f) = U_f. \notag
        \end{align}
\end{proposition}
From a physical point of view, \cref{eq:2} describes an Euler-Lagrange equation for the state of the system $U(t)$, \cref{eq:1} describes the dynamics of the adjoint state $\Lambda(t)$ and \cref{eq:3} imposes conditions on the control $u(t)$.
We can translate this formulation into a pseudo-Hamiltonian problem using the Pontryagin pseudo-Hamiltonian defined by:
\begin{align*}
    H_P&=\hbar \Lambda \cdot \dot{U} - \mathcal{L} = \Lambda \cdot F - F_0 \;,
\end{align*}
where $\cdot$ denotes the real Hilbert–Schmidt inner product of matrices, defined as $A \cdot B = \operatorname{Re} \operatorname{Tr}(B^\dagger A)$ (the real part is necessary for the pseudo-Hamiltonian obtained from the PMP.). Given that the system starts from the identity and follows a physical trajectory where \cref{eq:2} is satisfied, a solution to the optimization problem must satisfy the following necessary conditions
\begin{align}\label{eq:PMP_Hamiltonian}
    H_P&=\Lambda \cdot F - F_0,\\
    \hbar \dot{\Lambda} &= -\frac{\partial H_P}{\partial U}, \label{eq:lambdadot}\\
    \hbar \dot{U} &= \frac{\partial H_P}{\partial \Lambda}, \label{eq:UdotisF}\\
    \vec{0}&=\frac{\partial H_P}{\partial u},\label{eq:Denisuk}\\
    U(0) &= I, \quad U(t_f) = U_f. \label{eq:U0}
\end{align}

\subsubsection{Derivation of the parameterisation (MAGICARP)}
In the following we
work on the fixed computational time interval $s\in[0, 1]$. For a quantum evolution,~\cref{eq:2} corresponds to the Schrödinger
equation $\hbar\dot U=-iH(s)U$ where $U(0)$ in~\cref{eq:U0} is the identity and the system's Hamiltonian is
\begin{equation}\label{eq:Hkuk_time_s}
    H(s) = \hbar \Omega_0 \sum_{k=1}^{N_{\text{ctrls}}} u_k(s)H_k \;.
\end{equation}

Choosing the energy of the pulses, defined in \ref{sec:GRDvsGRAPE}, as the running cost
\begin{equation}\label{eq:running_cost_amplitude}
    F_0(u(s))= \hbar \Omega_0 \sqrt{\sum_k u_k^2(s)} \;,
\end{equation}
in the pseudo-Hamiltonian~\cref{eq:PMP_Hamiltonian}, Janković et al. derived in~\cite{Magicarp} the
following result using the PMP.
\begin{proposition}\label{prop:shape_controls}
    If the pulses $u_k(s)$ in~\cref{eq:Hkuk_time_s} minimize the running cost~\cref{eq:running_cost_amplitude}, then their shape is given by
    \begin{equation}
        \frac{u_k(s)}{c(s)} = \operatorname{Re} \operatorname{Tr} \left( U(s) M U^{\dagger}(s) H_k \right),
        \label{eq:optimal_control}
    \end{equation}
    where $U(s)$ is the propagator of the controlled system at time $s$, $M$ is a constant traceless Hermitian matrix and
    \begin{equation}\label{eq:enveloppe}
        c(s) = \sqrt{\sum_k u_k^2(s)}.
    \end{equation}
\end{proposition}
\begin{proof}
    In the pseudo-Hamiltonian $H_P$ of~\cref{eq:PMP_Hamiltonian}, by choosing $F = -\iu H U$ where $H$ the Hamiltonian of the controlled system in~\cref{eq:Hkuk_time_s}, we recover Schrödinger's equation on the propagator from~\cref{eq:UdotisF}, that is $\hbar \dot{U} = -i H U$. Given the running cost~\cref{eq:running_cost_amplitude}, the pseudo-Hamiltonian rewrites
    \begin{align}\label{eq:HP_with_F0}
        H_P&= \operatorname{Re} \operatorname{Tr} \left( -i \Lambda^\dagger H U \right) - \hbar \Omega_0\sqrt{\sum_k {u_k^2}}.
    \end{align}
    Substituting this expression for $H_P$ in~\cref{eq:lambdadot} on $\Lambda$ and in the maximization condition~\cref{eq:Denisuk}, we obtain the following conditions:
    \begin{align}\label{eq:PMP_with_F0}
        \hbar\dot{\Lambda} &= -i H \Lambda,\\
        \hbar\dot{U} &= -iHU,\label{eq:lambdadot2}\\
        \vec{0} &= \operatorname{Re} \operatorname{Tr}  \left(-i \Lambda^\dagger\frac{\partial H}{\partial\vec{u}}U\right) - \hbar \Omega_0\frac{\vec{u}}{\sqrt{\sum_k u_k^2}},\label{eq:magi_controls}\\
        U(0) &= I, \quad U(1) = U_f \notag \;.
    \end{align}
    A useful observation is that equation~\cref{eq:lambdadot2} is can be solved for $\Lambda$ in terms of the propagator $U$. Indeed, differentiating the function $\Lambda(s)^\dagger U(s)$, we find that its derivative is zero everywhere. Hence $\Lambda(s) = U(s)\Lambda(0)$. Since $\Lambda(0)$ is in the Lie algebra of the special unitary group, $\Lambda(0) = - i M$ for some traceless Hermitian matrix $M$. We then rewrite $-i \Lambda^\dagger$ in~\cref{eq:magi_controls} as $M U^\dagger$. Using the cyclic invariance of the trace and $\partial H/\partial u_k=\hbar \Omega_0 H_k$ we obtain the desired expression for the control pulses given in~\cref{eq:optimal_control}.
\end{proof}

MAGICARP is a parametrization of controls by means of the matrix $M$ representing the initial adjoint state. By means of a suitable choice of optimization routine, starting from an initial guess, this parameter is iteratively updated until the target gate is attained up to some user-defined tolerance.
As such, it belongs to the class of shooting methods in which, roughly, one aims at finding the appropriate initial covector in such a way that the set of necessary conditions derived from the PMP are satisfied and the desired target state is attained with sufficient accuracy. In particular MAGICARP suffers from the necessity of finding relevant initial guesses for the parameter $M$.
However what is worth noting here is that, by working with the propagator of the equation of motion and exploiting the Lie group structure, the maximization condition of the PMP leads to a control in closed-loop form. In other words, the shape of the controls are determined only by the instantaneous propagator $U(t)$.
One notable advantage is that the parameter $M$ is not tied to any choice of discretization scheme and fully determines the associated control given the control Hamiltonians.

\begin{remark}[Time optimality]
    Consider the energy running cost $F_0$ in~\cref{eq:running_cost_amplitude}. Minimising $\int_0^1 F_0(u(s))\,\mathrm{d}s$ is simply the time-minimal problem reformulated in a normalized computational variable $s$. Indeed for driftless dynamics, minimizing the physical duration $t_f$ under the amplitude constraint
    $\sqrt{\sum_{k}u_k^2(t)} \leq \Omega/\Omega_0$ is equivalent to minimizing the pulse area
    $\hbar\Omega_0 \int_0^1\sqrt{\sum_{k}u_k^2(s)}\,\mathrm{d} s$ on the fixed computational time interval $0 \leq s \leq 1$. This is similar in spirit to the rescaling of the time duration of pulses in \ref{sec:GRDvsGRAPE} to match (and thus maximize) the experimentally allowed energy.
\end{remark}

\begin{remark}[Normal extremals]\label{remark:abnormal}
    Proposition~\ref{prop:shape_controls} applies only to the so-called \emph{normal} extremals of the Pontryagin maximum principle. In other words, for certain choices of control Hamiltonians and target gates, the optimal controls may not be representable by~\eqref{eq:optimal_control} for any traceless Hermitian matrix M. Nevertheless, this simple parametrization remains a natural candidate for numerical exploration in the search for more time-efficient control pulses. See remark~\ref{rem:abnormal_extremals} in appendix~\ref{app:Numerical_framework} for further discussion.
\end{remark}

\begin{remark}[Necessity]
    We emphasize that the shape of the controls in Proposition~\ref{prop:shape_controls} is only necessary for optimality. Put differently and neglecting the subtleties of Remark~\ref{remark:abnormal}, if a pulse is optimal then it has this shape but a pulse with this shape need not be optimal.
\end{remark}
\section{Numerical proof-of-concept study on a realistic physical system}~\label{sec:Numerical_implementation_on_SMM}
\subsection{Single-Molecule Magnets}~\label{subsec:SMM}
The interest in using single-molecule magnets (SMMs) as quantum devices was initiated in the 90's, with the first observation of quantum tunneling of the magnetisation (QTM) transitions in Mn$_{12}$~\cite{SMM_concept_discovery, SMMs_in_spintronics_plus_history}. Since then, molecules have been engineered that amplify these effects and exhibit long spin and coherence lifetimes, most notably the terbium double (and triple) deckers. We will compare the numerical results obtained using GRAPE and MAGICARP on these two molecules.
\subsubsection{Physical and experimental constraints}
The first molecule of interest is the bisphthalocyaninato terbium(III) ($\text{Tb}\text{Pc}_2$) molecule, also known as the double decker, a SMM synthesized in 2003~\cite{TbPc2_discovery_as_SMM} and shown in \cref{fig:double_decker}. The $\text{Tb}^{3+}$ ion has a nuclear spin of $I=3/2$~\cite{godfr2017,janko2024}, which gives rise to $2I + 1 =4$ energy levels. This makes the $\text{Tb}\text{Pc}_2$ an effective qudit with $d = 4$ levels~\cite{qudit_concept}. It is also possible to combine two double deckers into a triple decker $\text{Tb}_2\text{Pc}_3$ molecule which then, thanks to a coupling between the two $\text{Tb}^{3+}$ ions~\cite{triple_deckers}, can act as an effective qudit with $d = 16$ levels~\cite{Tb2Pc3_expe}.
\begin{figure}[hbt]
    \centering
    \subfloat[\label{fig:double_decker}]{%
        \includegraphics[width=0.48\textwidth]{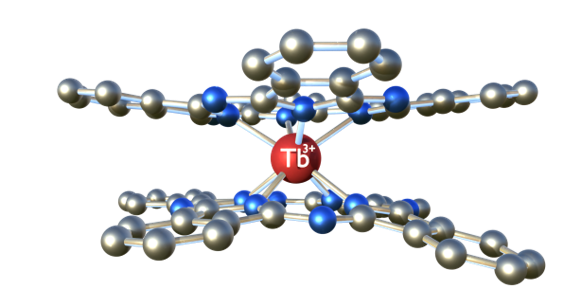}%
    }
    \hfill
    \subfloat[\label{fig:triple_decker}]{%
        \includegraphics[width=0.48\textwidth]{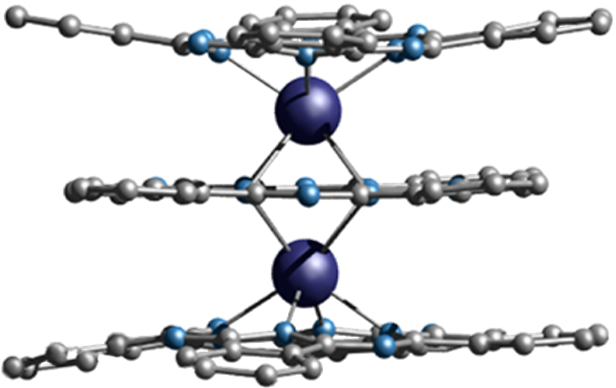}%
    }
    \caption{(a) A $\text{Tb}\text{Pc}_2$ double decker, composed of the central $\text{Tb}^{3+}$ ion (red) surrounded by nitrogen (blue) and carbon (grey) atoms of the phthalocynanine ligands. Image from~\cite{hartm2024}.\ (b) A triple decker, $\text{Tb}_2\text{Pc}_3$. The two dark blue atoms are also $\text{Tb}^{3+}$ ions. Image from~\cite{hartm2024}.}
\label{fig:chemical_structure_TbPc}
\end{figure}
To construct quantum gates we need to trigger transitions between the system's energy levels by applying an electromagnetic pulse with an appropriate frequency. Not all transitions are allowed. In the example of the double decker, selection rules only allow for addressing transitions between levels separated by a $\Delta I = \pm 1$~\cite{janko2024}. 
\begin{figure}[hbt] 
    \centering
    \subfloat[\label{fig:double_levels}]{%
        \includegraphics[width=0.48\textwidth]{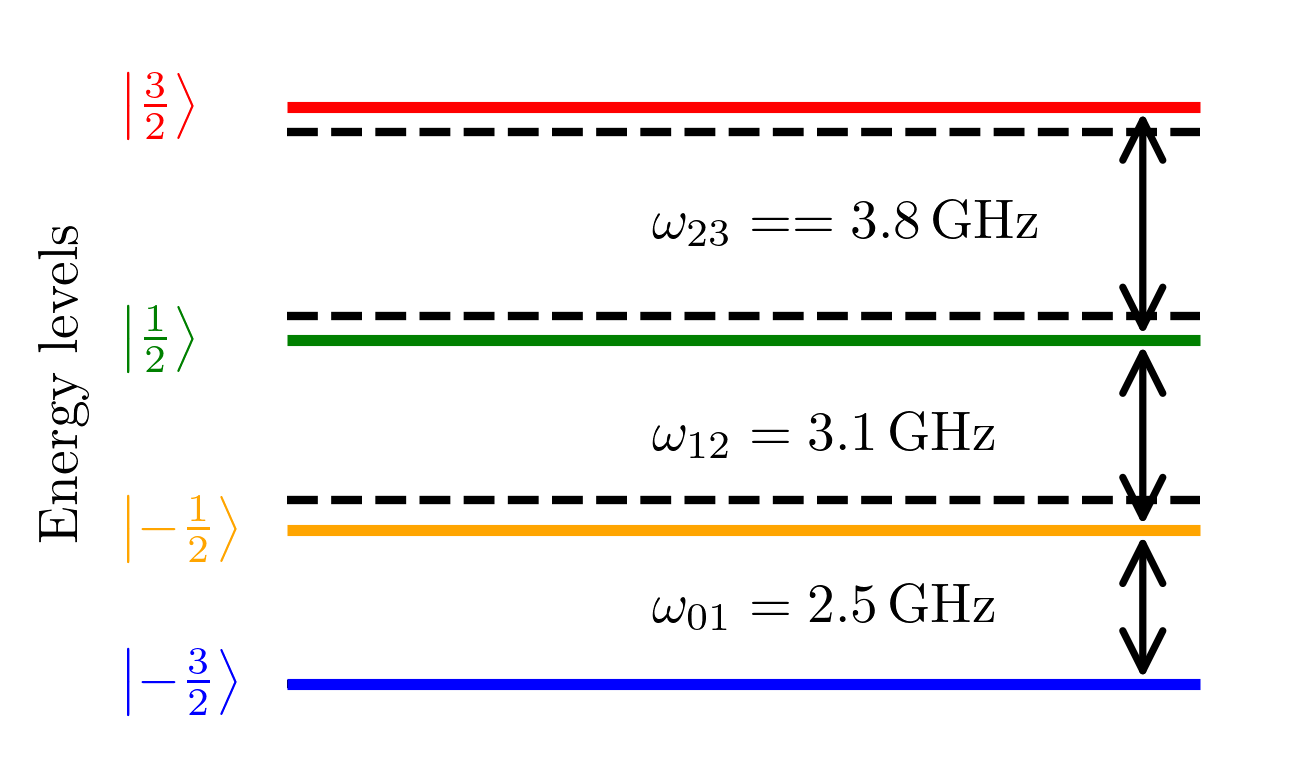}%
    }
    \hfill
    \subfloat[\label{fig:graph}]{%
        \includegraphics[width=0.48\textwidth]{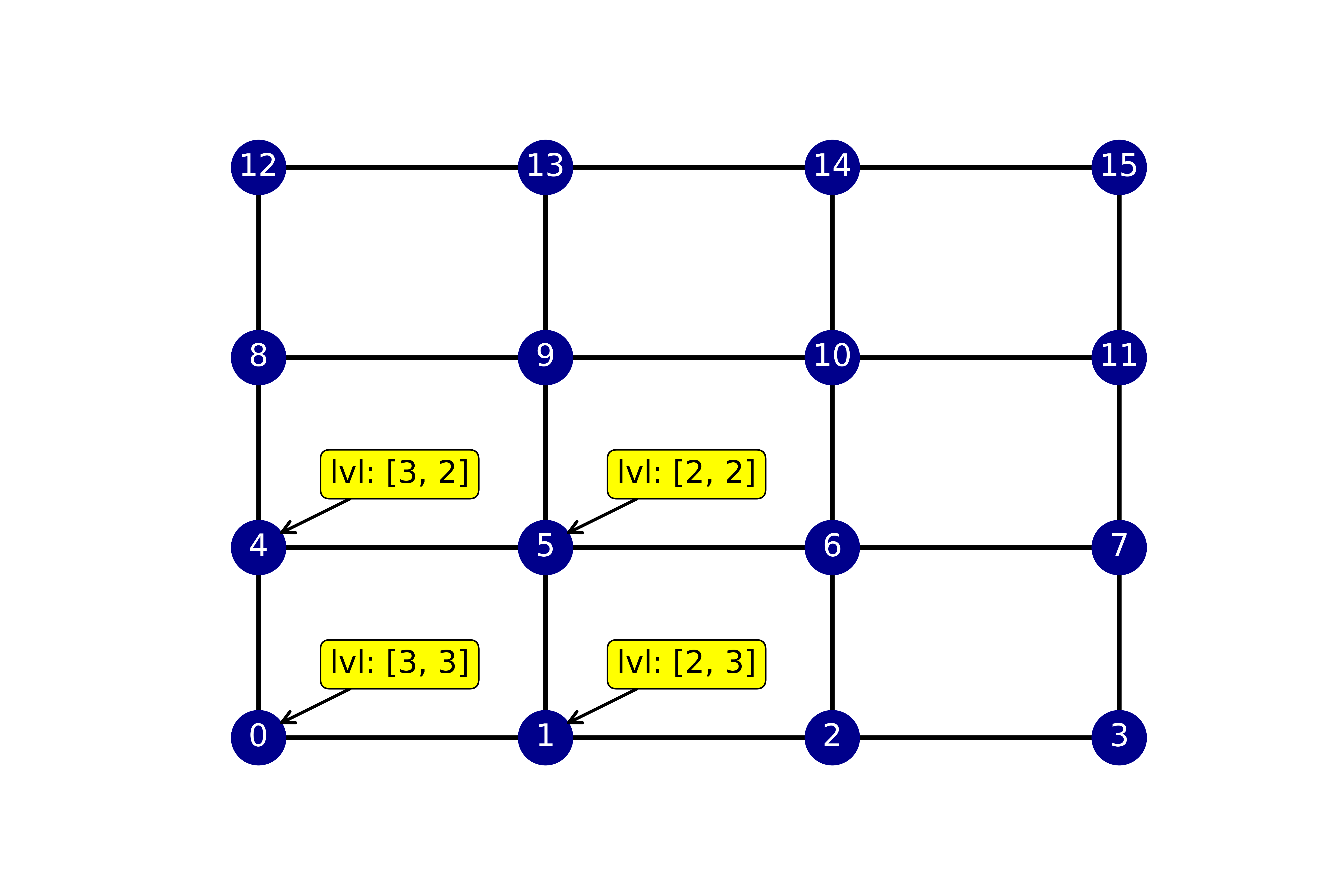}%
    }
    \caption{(a) The linear graph connecting the four energy levels of a double decker $\text{Tb}\text{Pc}_2$. The spacing is unequal thanks to a non linear term in the Hamiltonian (electric quadrupole hyperfine interaction). Otherwise energy levels would be positioned on the horizontal dashed lines. Values are taken  from~\cite{godfr2017}.\ (b) The graph connecting the sixteen energy levels of a triple decker $\text{Tb}_2\text{Pc}_3$. Each node represent an energy level and they are ranked by their energy. The edges represent addressable Rabi transitions between nodes. The energy levels are labeled as ``lvl'', where the first digit denotes the hyperfine level of the first ion Tb$^{3+}$ in Tb$_{2}$Pc$_{3}$, and the second digit denotes the hyperfine level corresponding to the other Tb$^{3+}$ ion (see fig. \ref{fig:chemical_structure_TbPc}). This figure is done using parameters that differ slightly from~\cite{Tb2Pc3_expe}, and at a large magnetic field in order to suppress level mixing.}\label{fig:deckers_graph}
\end{figure}
This means that only transitions shown in~\cref{fig:double_levels} have an effect on the system. For instance, an energy level corresponding to $I=-1/2$ can communicate when excited at the right frequency with the energy levels corresponding to $I=-3/2$ or $I=1/2$ but not with the level corresponding to $I=3/2$. Thus only neighbouring interactions are allowed, resulting in what we refer to as a linear graph. The resonances can be individually adressed thanks to a non-linear term that results from an electric quadrupole hyperfine interaction. For the case of the triple decker, only coupling between adjacent levels of nuclear spin $I_1$ and $I_2$ ($\Delta I = 1$) are allowed, the 1 or 2 indices denoting the hyperfine energy levels for each of the $\text{Tb}^{3+}$ ions. The non trivial resulting graph is shown in~\cref{fig:graph}. 
Usually the allowed transitions are determined experimentally and as long as all energy levels are connected, it is possible to create any unitary gate on the Hilbert space associated to these SMM, for instance by using the GRD method.
\subsubsection{Illustrations for a Quantum Fourier Transform}~\label{subsubsec:QFT_on_SMM}

The Quantum Fourier Transform (QFT) on a Hilbert space of dimension $d$ is a linear transformation that maps each computational basis state $\ket{x}$ to an equal-magnitude superposition of all basis states, weighted by complex phases:
\begin{equation*}
    \ket{x} \;\longrightarrow\; \frac{1}{\sqrt{d}} \sum_{k=0}^{d-1} e^{2\pi i x k / d} \ket{k}, \quad x = 0, 1, \dots, d-1.
\end{equation*}
\begin{figure}[tbp]
    \centering
    \subfloat[]{%
        \includegraphics[width=0.45\textwidth]{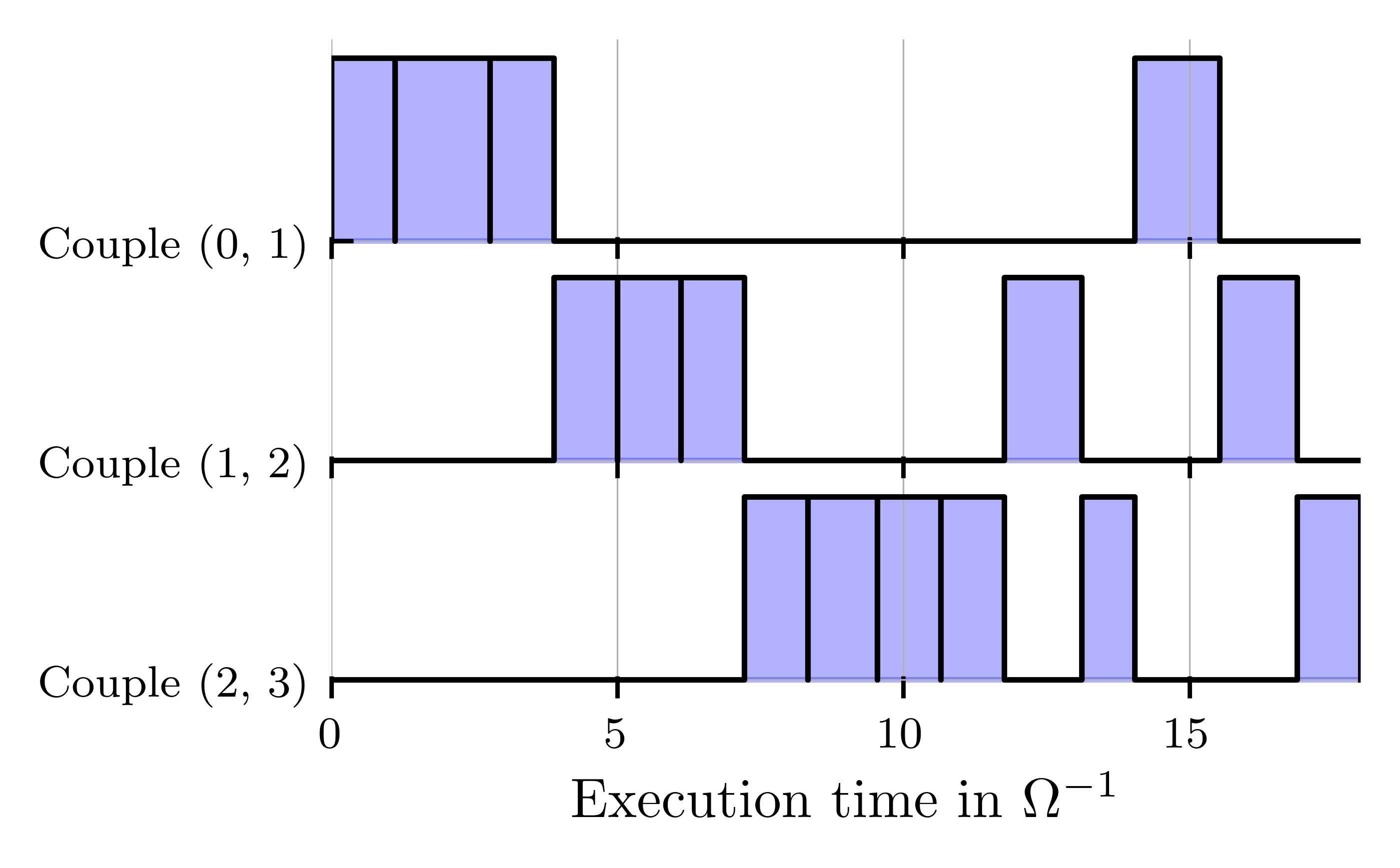}%
        \label{fig:controls_4}%
    }
    \hfill
    \subfloat[]{%
        \includegraphics[width=0.45\textwidth]{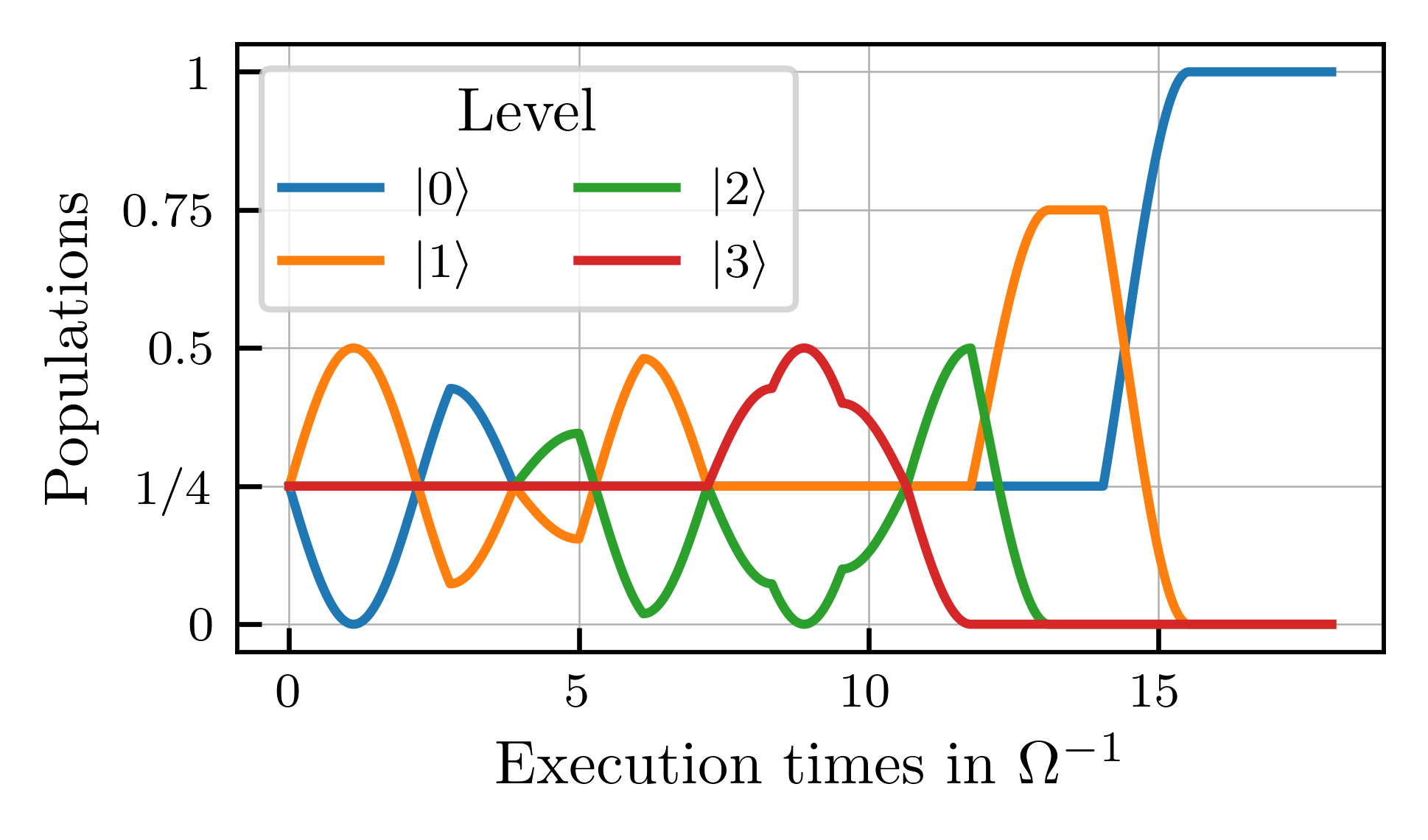}%
        \label{fig:pop_dynamics_4}%
    }

    \vspace{0.1cm}

    \subfloat[]{%
        \includegraphics[width=0.45\textwidth]{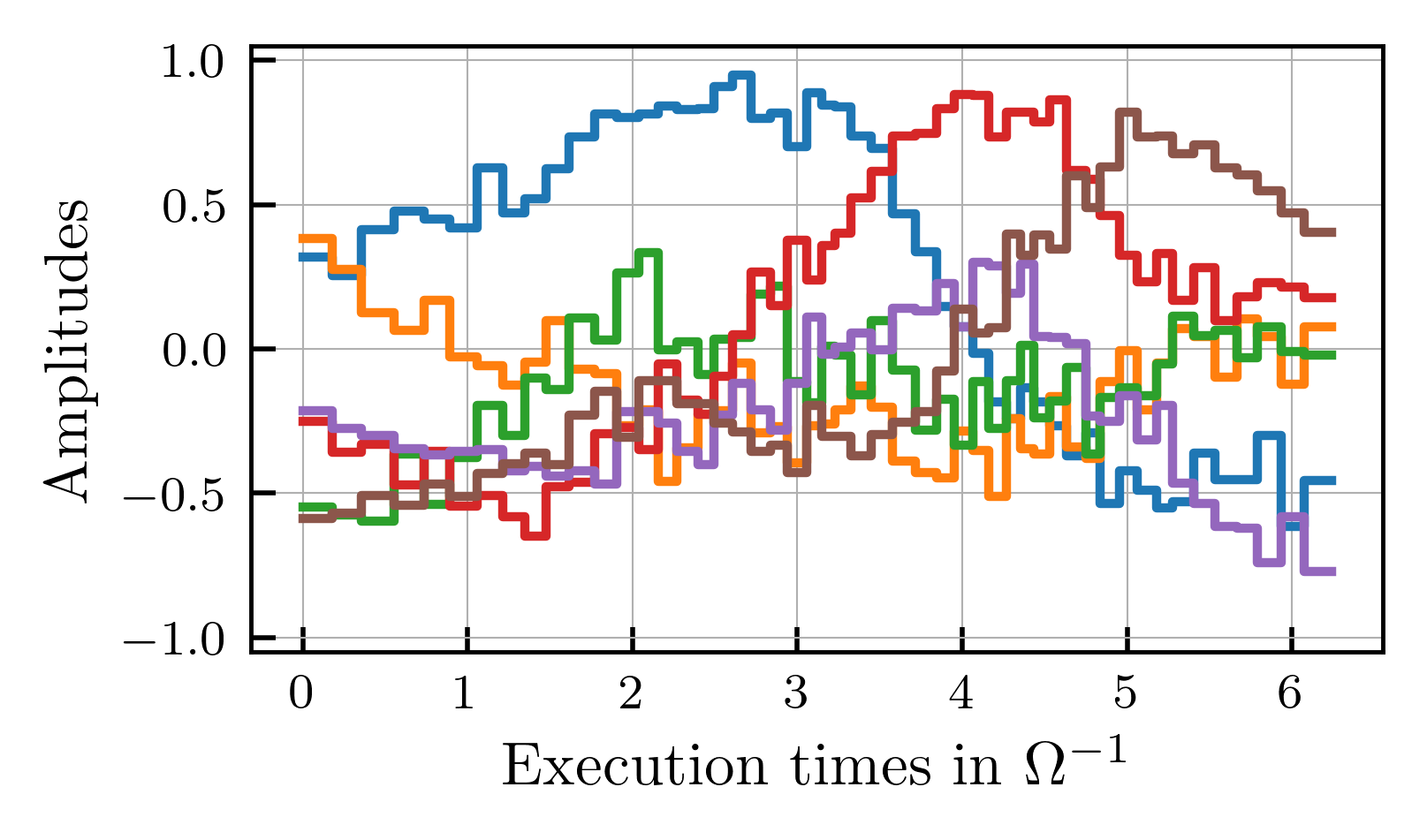}%
        \label{subfig:GRAPE_amps_H_4}%
    }
    \hfill
    \subfloat[]{%
        \includegraphics[width=0.45\textwidth]{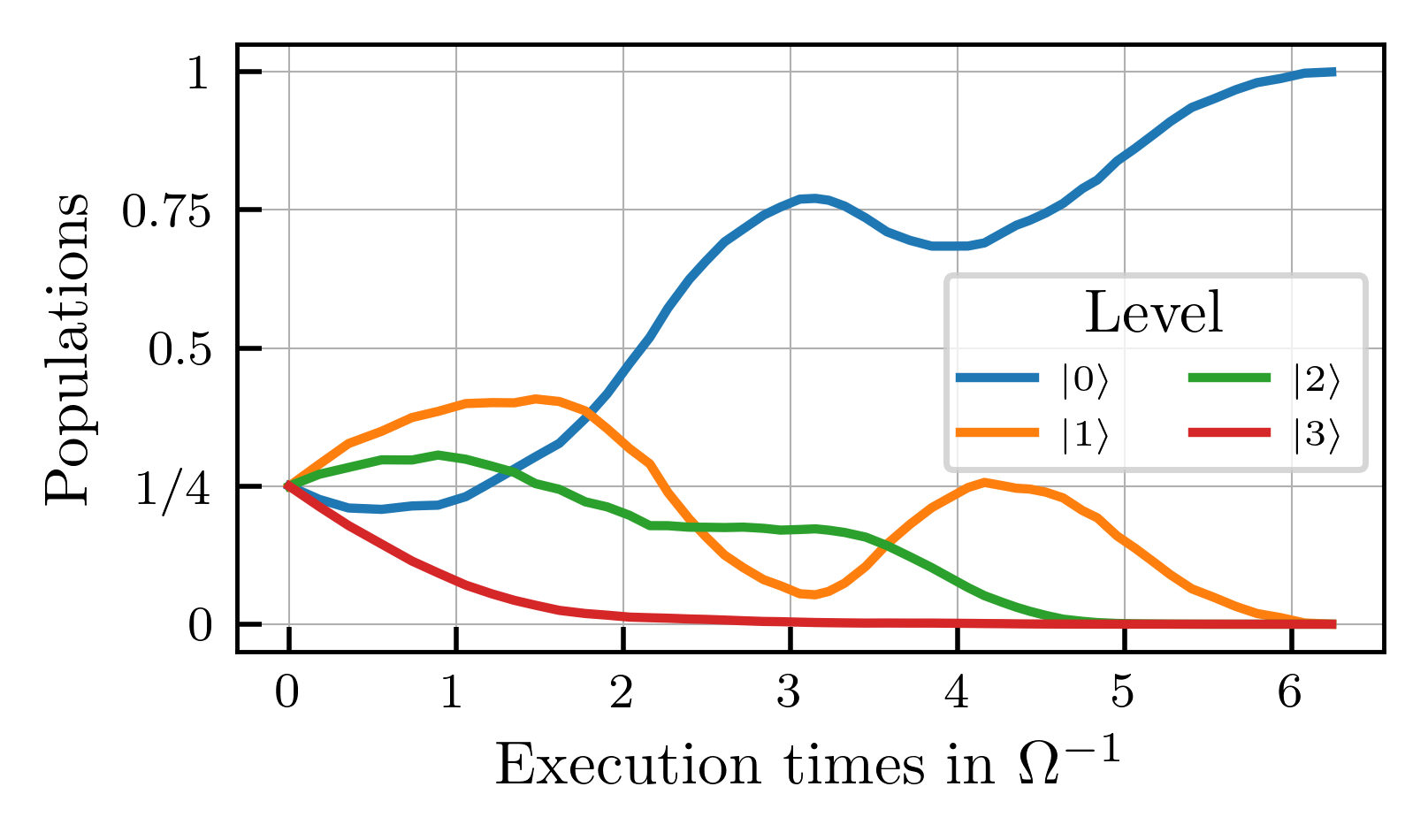}%
        \label{subfig:GRAPE_pop_H_4}%
    }

    \vspace{0.1cm}

    \subfloat[]{%
        \includegraphics[width=0.45\textwidth]{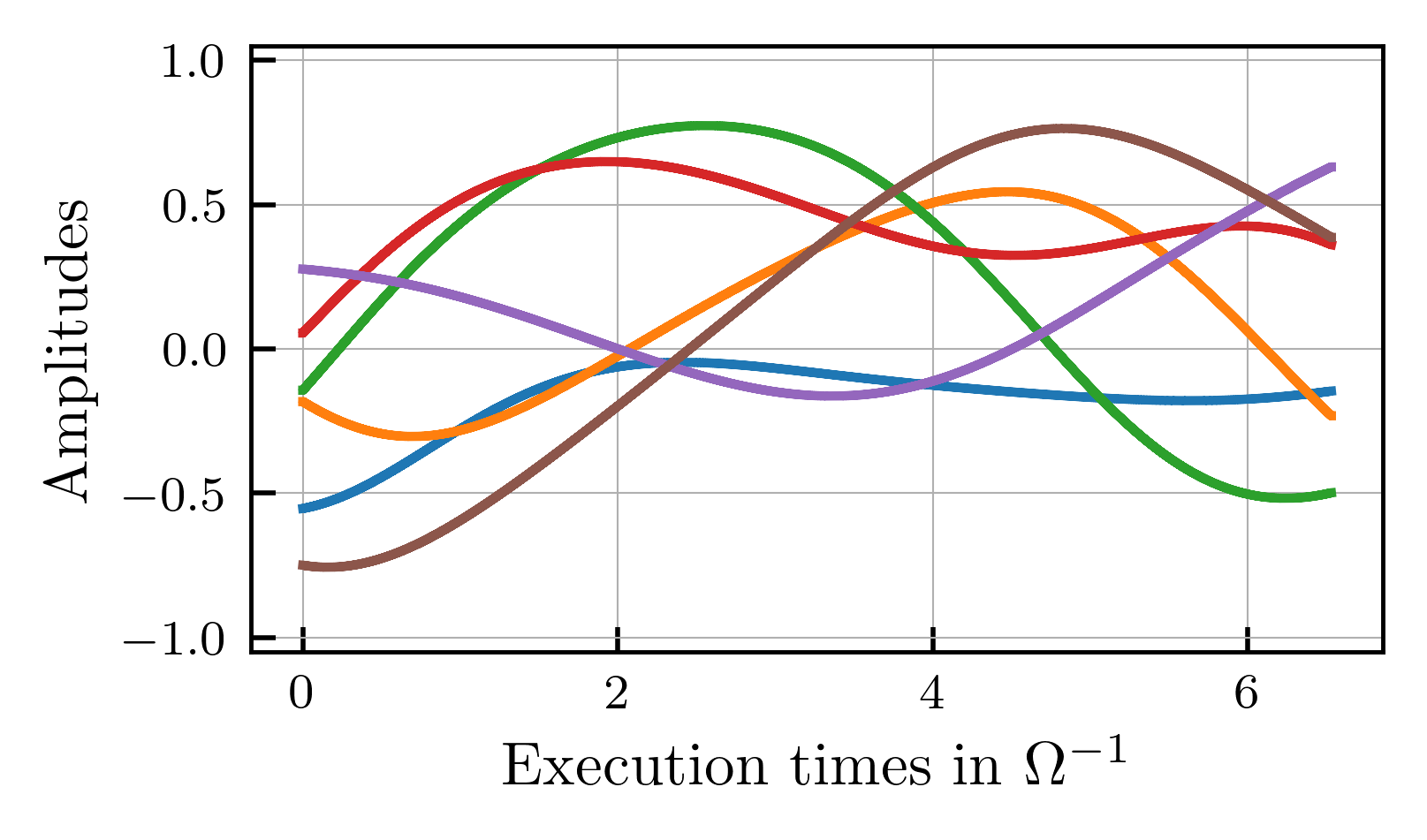}%
        \label{subfig:MAGICARP_amps_H_4}%
    }
    \hfill
    \subfloat[]{%
        \includegraphics[width=0.45\textwidth]{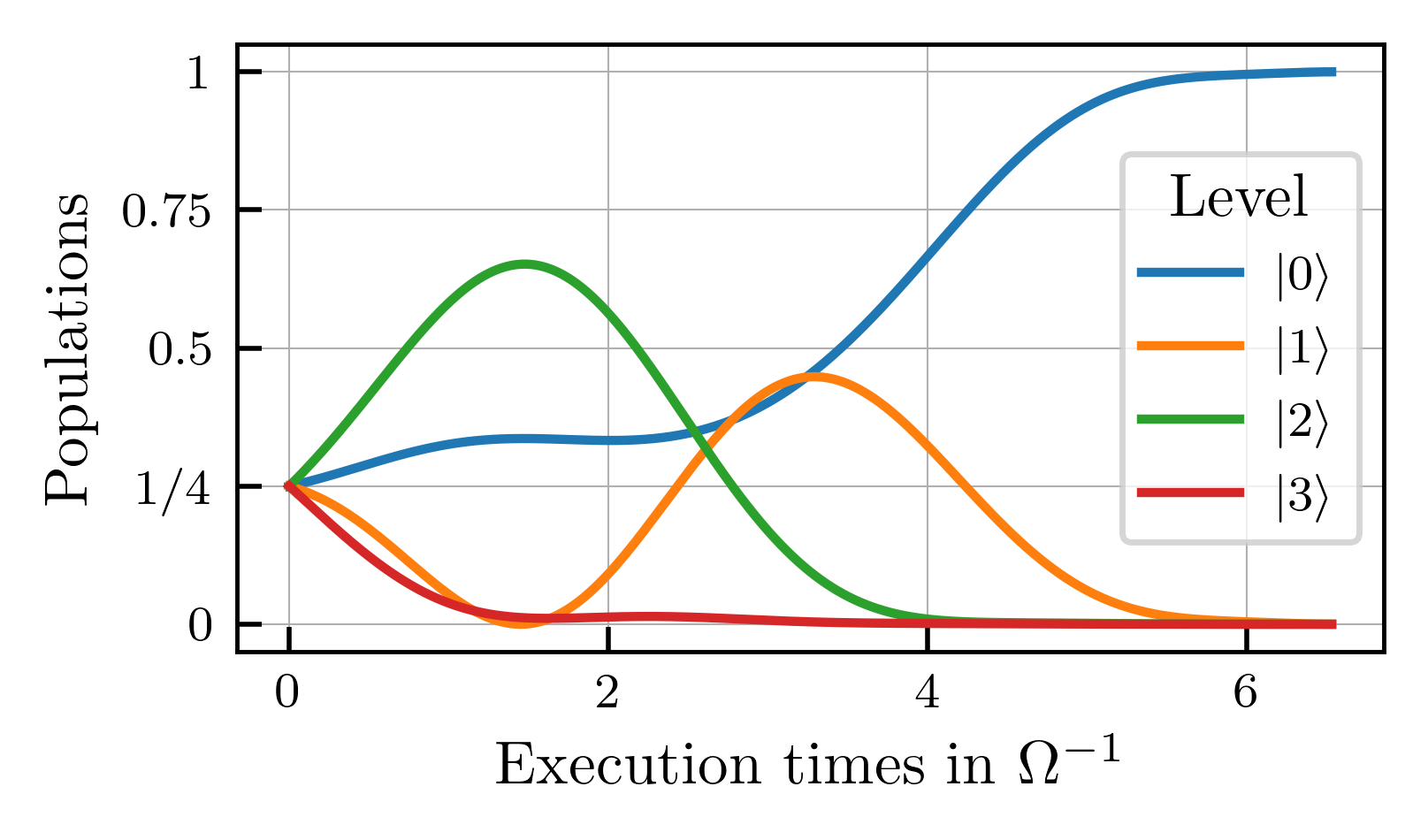}%
        \label{subfig:MAGICARP_pop_H_4}%
    }

    \caption{Control amplitudes (left panels) and population dynamics (right panels) for the implementation of a QFT gate on the double decker graph (see fig. \ref{fig:deckers_graph}a). Panels (a)-(b) correspond to GRD, (c)-(d) to GRAPE, and (e)-(f) to MAGICARP. In (a), each blue activation, delimited by black lines, correspond to a different $(\theta, \phi)$ for a GR in~\cref{eq:GRD}. The length of each blue activation is $\phi$. For GRAPE and MAGICARP, each color denotes a distinct control amplitude. The target fidelity is $10^{-4}$. GRAPE controls are discretized in 50 time steps. All times are given in units of $\Omega^{-1}$ (see \cref{sec:GRDvsGRAPE}).\label{fig:controls_and_pop}}
\end{figure}
The QFT is a key gate for quantum computation, including Shor's algorithm and many other important algorithms, because it is often a good starting point for solving problems in the hidden subgroup problem~\cite[5.4.3]{nielsen00}, a class of problems for which quantum computations can provide exponential speed-ups. It is also referred to in the literature as a generalized Hadamard gate on qudits.
To illustrate the shape of control pulses implementing the QFT for the double-decker, we show in~\cref{fig:controls_and_pop} the pulses calculated using the GRD, GRAPE and MAGICARP algorithms and the addressable energy difference between levels shown in \cref{fig:double_levels}. Since there is 3 edges on this linear graph with one $\sigma_x$ and one $\sigma_y$ control per edge, there are 6 controls in total. The initial state is $\ket{\psi}=\frac{1}{\sqrt{N}}\sum_{i=0}^{N-1}\ket{i}$, a superposition of all the states, so that the populations end up in the state $\ket{0}$. 

In~\cref{fig:controls_4}, notice that z-rotations are applied first to clear the diagonal elements in~\cref{eq:localphases} and that each such rotation is a block of 3 operations. Indeed, they are constructed using controls over the $x$ and $y$ directions and the relation $R_z(\phi) = R_y\left(\frac{\pi}{2}\right) R_x(\phi) R_y\left(-\frac{\pi}{2}\right)$. For the GRD, the same figures are available for the triple decker in~\cref{app:triple_decker_dynamics}. Since the figure of merit $\Omega T_2$ is over the thousands for the double decker~\cite[6.5]{godfr2017}, these pulses are well in the range of experimental setups. 
It is noteworthy that the controls in~\cref{subfig:MAGICARP_amps_H_4} are smoother than those in~\cref{subfig:GRAPE_amps_H_4}. Indeed, the MAGICARP parametrization always yields approximations to smooth controls while the GRAPE method uses by construction piecewise constant pulses which might well have discontinuities.

\subsection{Comparison of MAGICARP relative to GRAPE}~\label{subsec:Magicarp_results}
In this work, for both the GRAPE and MAGICARP methods, we set the target infidelity to $10^{-4}$. We explain in~\cref{sec:1e-4_good_choice_GRAPE} why this is a relevant choice. Instead of optimizing controls directly inside an infinite dimensional function space, the MAGICARP allows for the optimization a finite number of parameters, namely the coefficients of the matrix \( M \) in~\cref{eq:optimal_control}, which are not tied to any discretization scheme but nevertheless determine smooth control functions.
A relevant property of this method is that the associated controls satisfy a necessary (albeit not sufficient) condition for a time-optimal control.
\begin{figure}[tbp]
    \centering
    \subfloat[]{
    \includegraphics[width=0.45\textwidth]{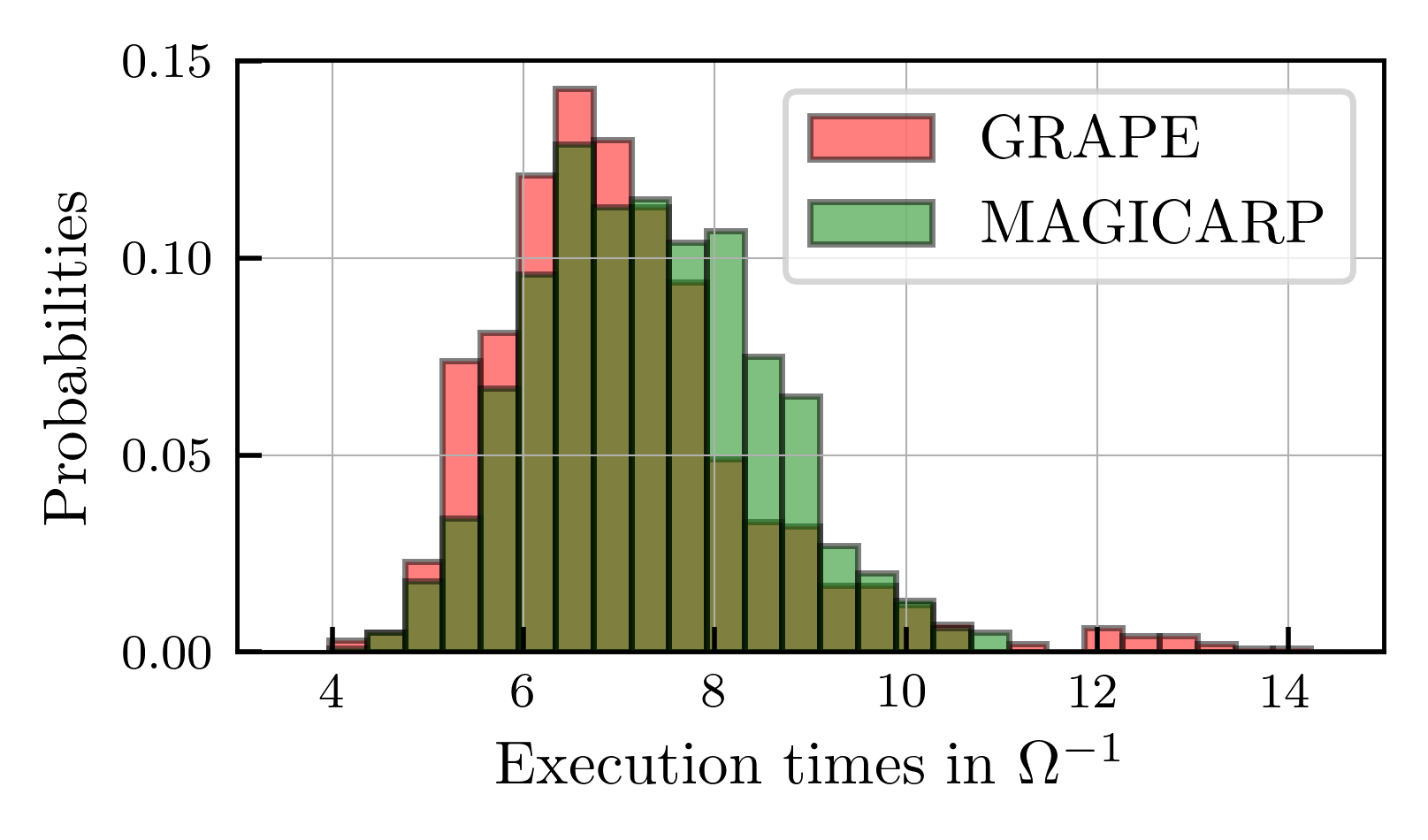}%
    \label{subfig:GRAPE-vs-MAGICARP-double}%
    }
    \hfill
    \subfloat[]{
    \includegraphics[width=0.45\textwidth]{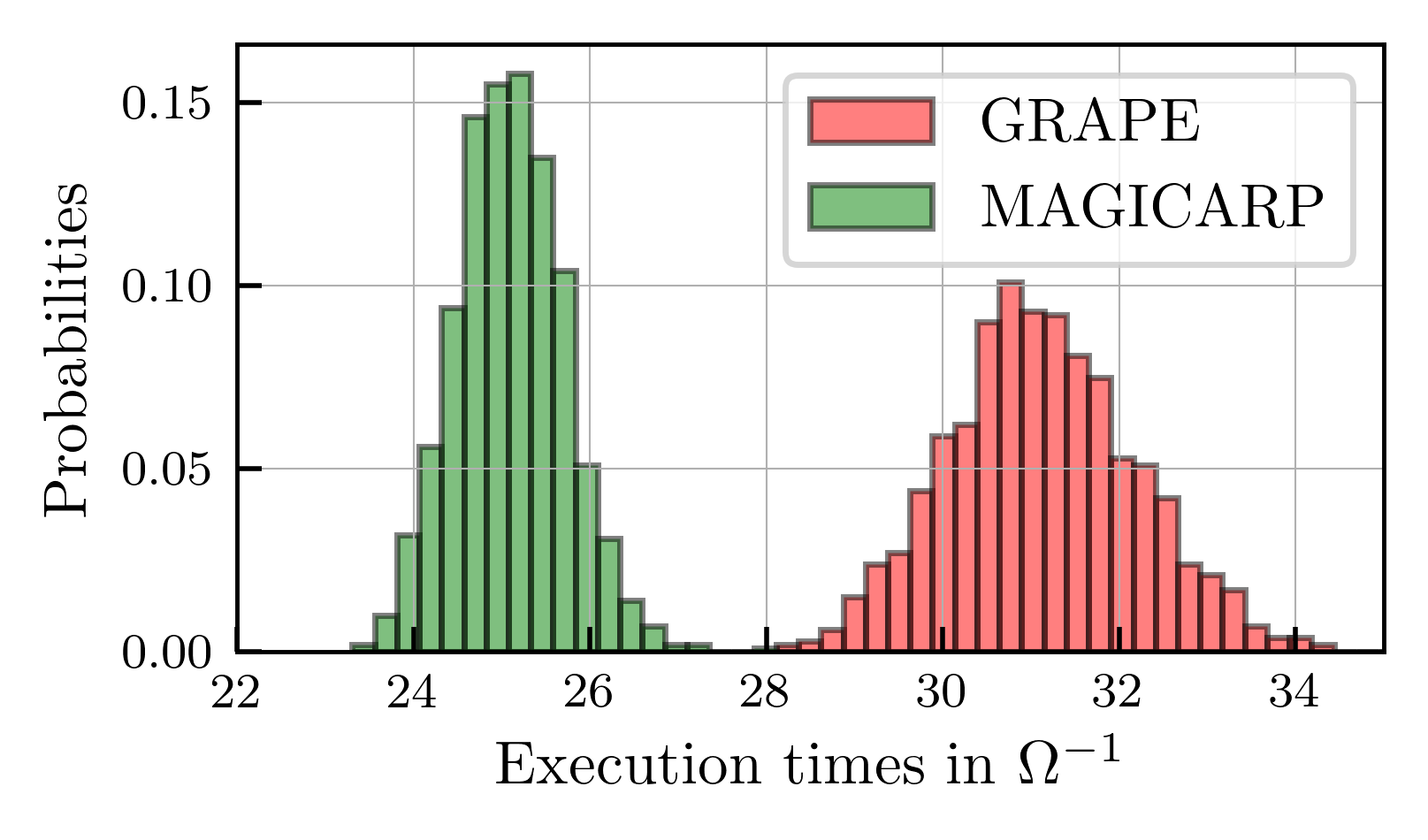}%
    \label{subfig:GRAPE-vs-MAGICARP-triple}%
    }
    \caption{Comparison of GRAPE and MAGICARP execution times to implement 1000 random Haar unitaries.
    (a) For a $\text{Tb}\text{Pc}_2$ molecule.
    (b) For a $\text{Tb}_2\text{Pc}_3$ molecule.}\label{fig:GRAPE-vs-GRD}
\end{figure}

From a pragmatic standpoint, it appears numerically that this parametrization yields shorter execution times than GRAPE, especially as the dimension of the Hilbert space increases. Indeed, for the triple decker in~\cref{subfig:GRAPE-vs-MAGICARP-triple}, the median execution time found by MAGICARP is about 80\% of that found by GRAPE. However, an important point is that numerically optimizing the matrix $M$ in~\cref{eq:optimal_control} is all the more difficult as the size of the Hilbert space grows. In~\cref{app:Numerical_framework}, we discuss an optimizer based on a natural gradient descent which provides solutions below the infidelity threshold for qudits up to $d = 16$ dimensions, e.g.~the triple decker discussed in~\cref{subsec:SMM}. We used this optimizer to obtain the results of~\cref{subfig:GRAPE-vs-MAGICARP-triple}. Codes written in Julia and Python are available on the Github repositories \href{https://github.com/killianlutz/pyMagicarp}{\texttt{Magicarp}} and \href{https://github.com/killianlutz/pyMagicarp}{\texttt{pyMagicarp}}.

Overall, the higher the dimension of the system is, the better the results of the MAGICARP in comparison to GRAPE when comparing median execution times for random gates or the minimum execution time for a specific gate (the QFT) as shown in~\cref{fig:Magicarp_dimensionality}. 
\begin{figure}[tbp]
    \centering
    \subfloat[]{
    \includegraphics[width=0.45\textwidth]{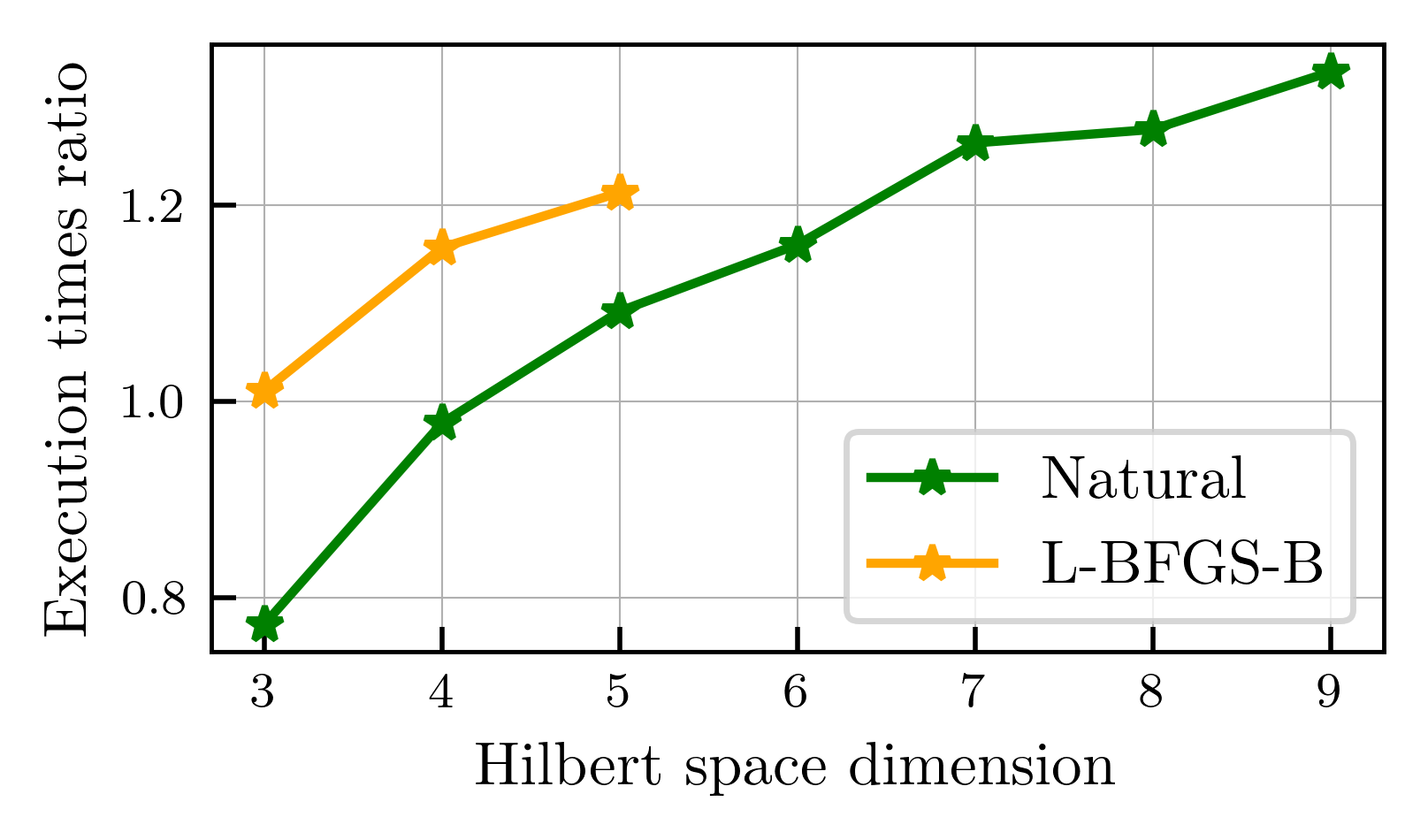}%
    \label{subfig:GRAPE-vs-MAGICARP-median-random}%
    }
    \hfill
    \subfloat[]{
    \includegraphics[width=0.45\textwidth]{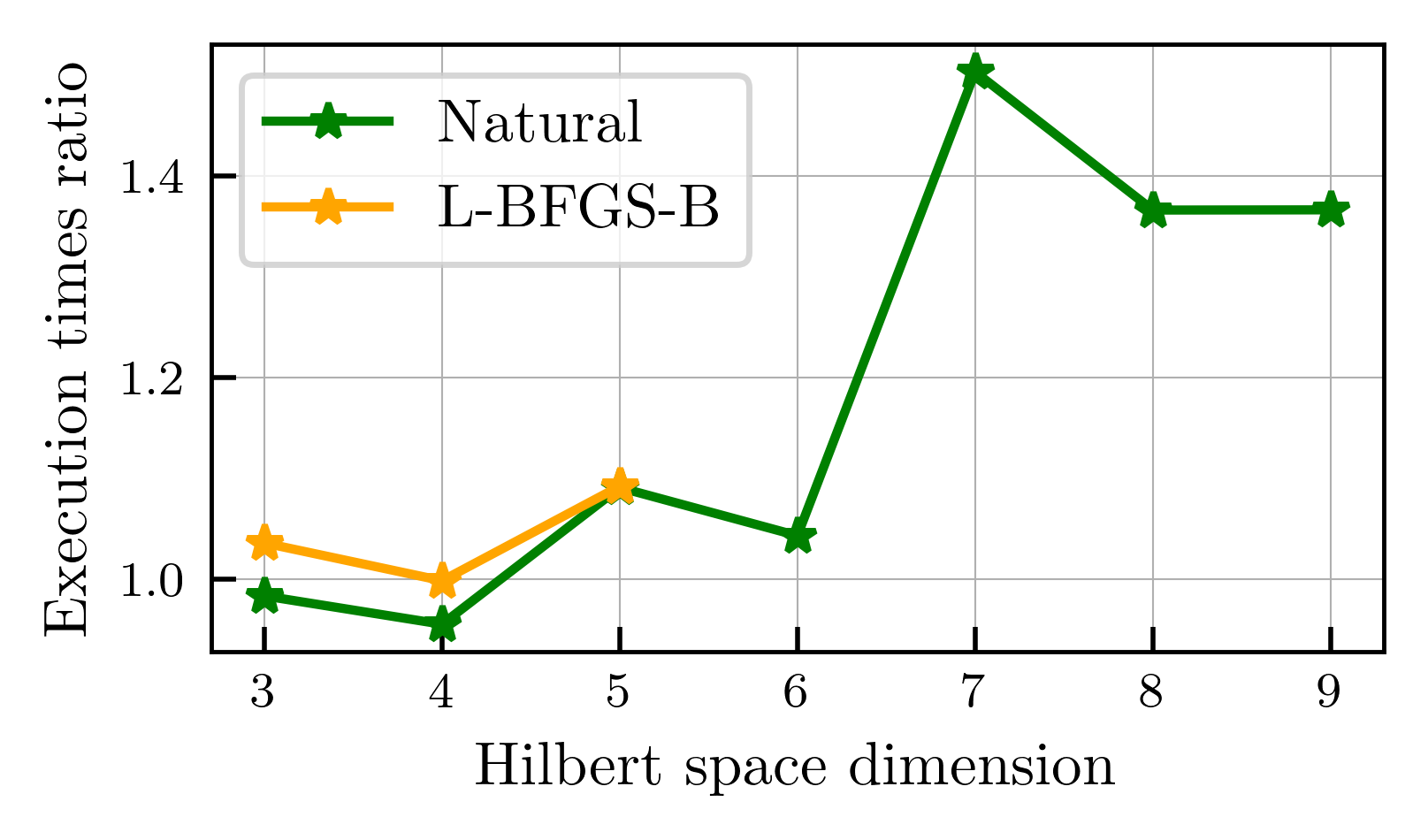}%
    \label{subfig:GRAPE-vs-MAGICARP-min-QFT}%
    }
    \caption{Comparison of GRAPE and MAGICARP execution times on linear graphs with 3 to 9 nodes, as a function of the Hilbert space dimension. There is at least 100 initializations for each dimension. We used the natural gradient descent described in~\cref{app:Numerical_framework} for the green curves, and for systems of dimension below 6, we also used the more computationally expensive L-BFGS-B optimizer from Scipy~\cite{SciPy} described by the orange curve. The ratios indicate how much faster MAGICARP is compared to GRAPE for a given dimension.\
    (a) Ratio of the median execution times for random Haar gates using GRAPE with random initializations over the median time using MAGICARP.\
    (b) Ratio of the minimum execution times for a QFT gate using GRAPE with random initializations over the minimum time using MAGICARP.}\label{fig:Magicarp_dimensionality}
\end{figure}
We choose to use a QFT gate to compare the minimum execution time for the MAGICARP and the GRAPE method for three reasons: (i) this gate is not trivial to generate given the experimental constraints, (ii) is key to several important quantum algorithms as explained in~\cref{subsubsec:QFT_on_SMM}, and (iii) has execution times similar to a random gate (gaussian-like) as one can see when comparing~\cref{subfig:GRAPE-vs-MAGICARP-triple} and~\cref{subfig:GRAPE-vs-MAGICARP-H}. 

We observe for low-dimensional systems ($d = 3$ and 4) that the natural gradient optimizer does not yield faster pulses since the ratios of execution times are below 1 in~\cref{fig:Magicarp_dimensionality} for the corresponding curve. To assess whether this observation should be attributed to the MAGICARP parametrization itself or to the optimizer we chose, we optimize once more the matrix $M$ using the same optimizer and hyperparameters as GRAPE (L-BFGS-B method from SciPy~\cite{L-BFGS-B-Scipy, SciPy}). By doing so, we observe that MAGICARP is at least faster than GRAPE for systems of dimension $d \leq 5$, both in median time for random gates and in minimum time for the QFT gate. This observation illustrates that the advantage of MAGICARP over GRAPE discussed in this paper leaves room for further improvements of the execution time by carefully designing the algorithm used to optimize the parameter of MAGICARP.

Comparing these methods only by implementing target gates chosen uniformly at random in the unitary group provides insights into some kind of improvement of the execution time ``on average''. However, in practice, not all unitary matrices are of interest for the existing quantum algorithms. We therefore compare the GRAPE and MAGICARP methods on the triple decker ($d = 16$) by implementing four commonly used quantum gates: QFT, T, X and SUMX~\cite{qudit_concept}.\ The latter generalizes the CNOT gate to qudits.
We obtained significantly better results in~\cref{fig:GRAPE-vs-GRD-specific} for the MAGICARP, regarding minimum execution times. In fact, for the QFT gate in~\cref{subfig:GRAPE-vs-MAGICARP-H}, the smallest execution time found by MAGICARP is about 74\% of the time found by GRAPE.
\begin{figure}[tbp]
\centering
\subfloat[]{
\includegraphics[width=0.45\textwidth]{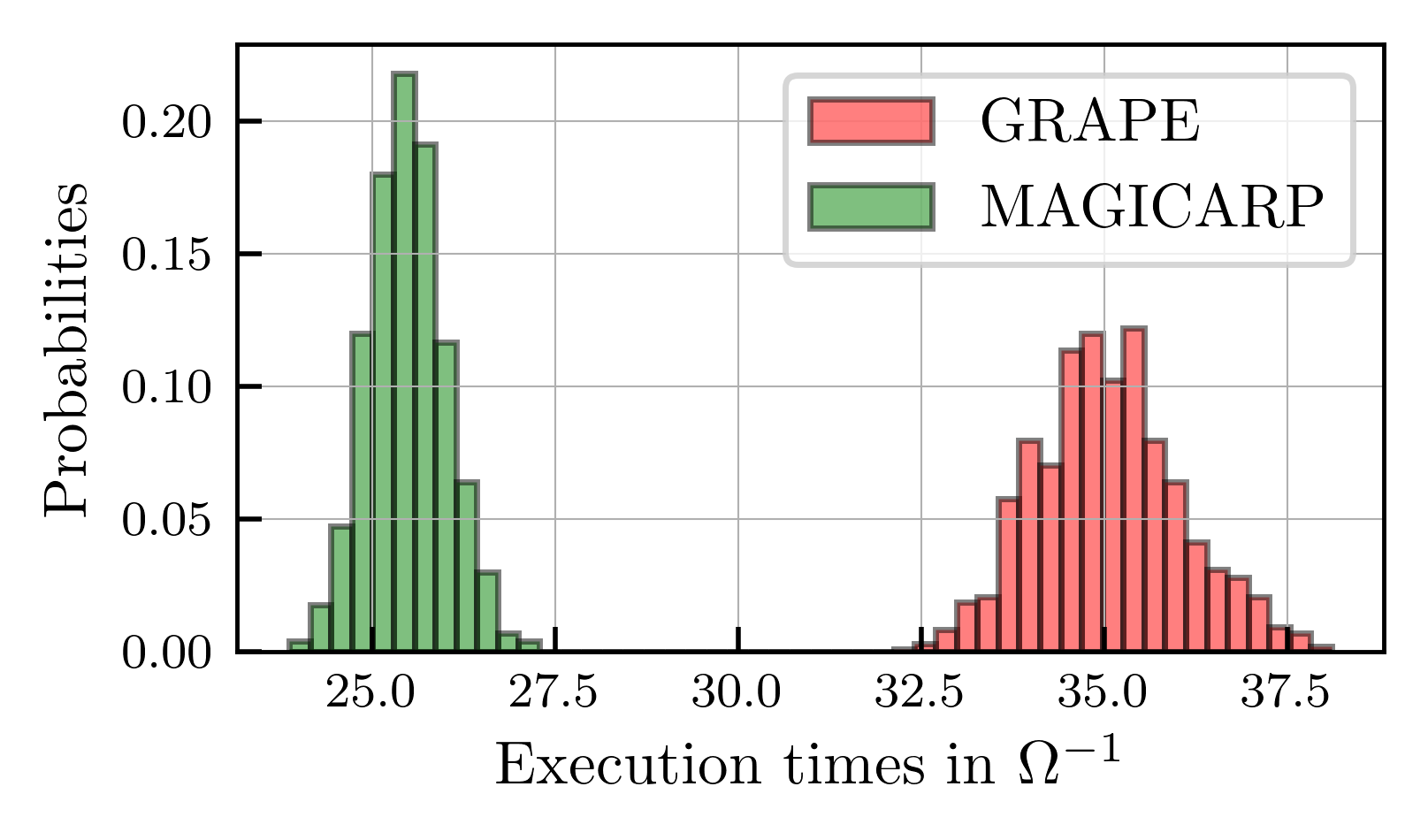}%
\label{subfig:GRAPE-vs-MAGICARP-H}%
}
\hfill
\subfloat[]{
\includegraphics[width=0.45\textwidth]{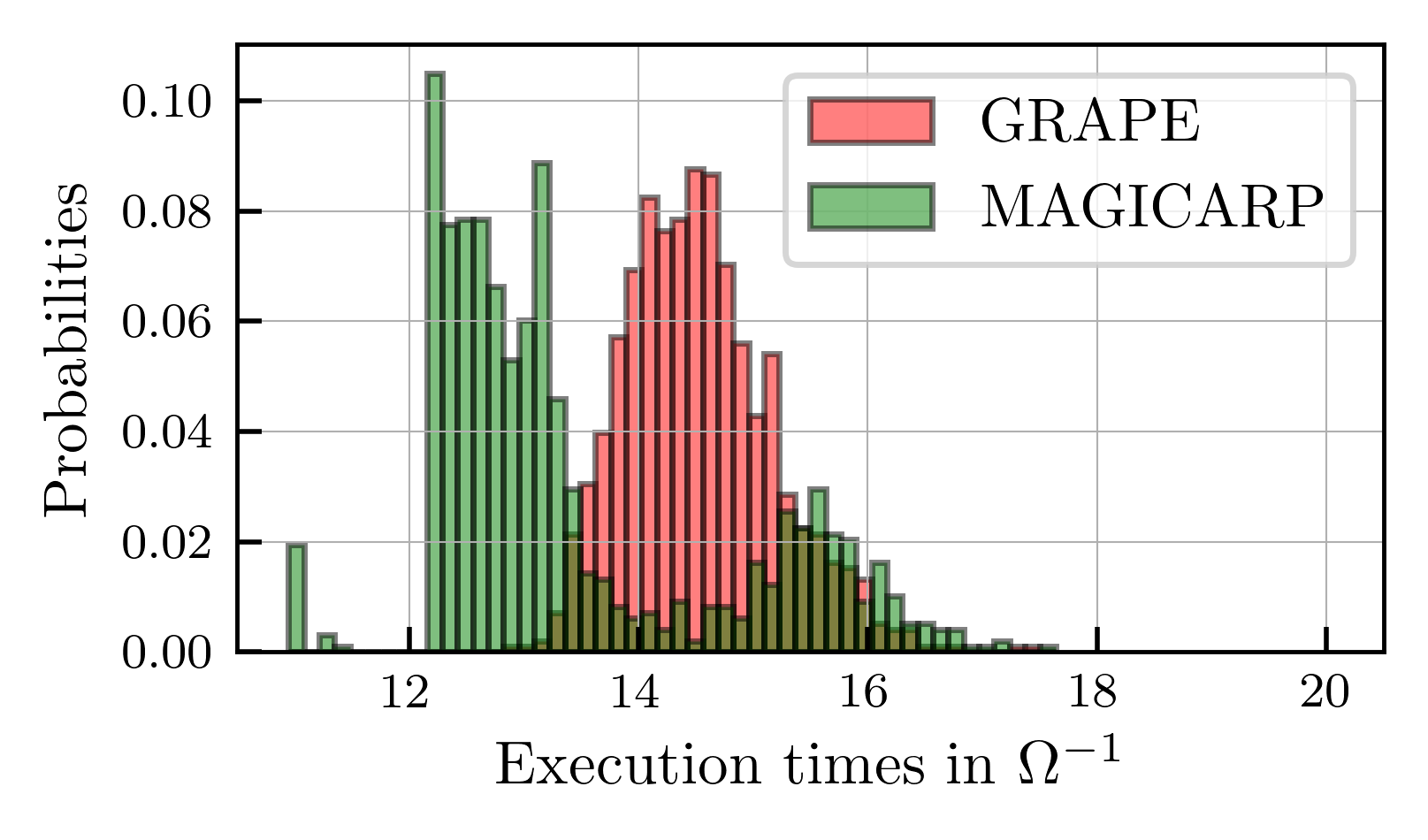}%
\label{subfig:GRAPE-vs-MAGICARP-T}%
}
\vspace{0.5em}
\subfloat[]{
\includegraphics[width=0.45\textwidth]{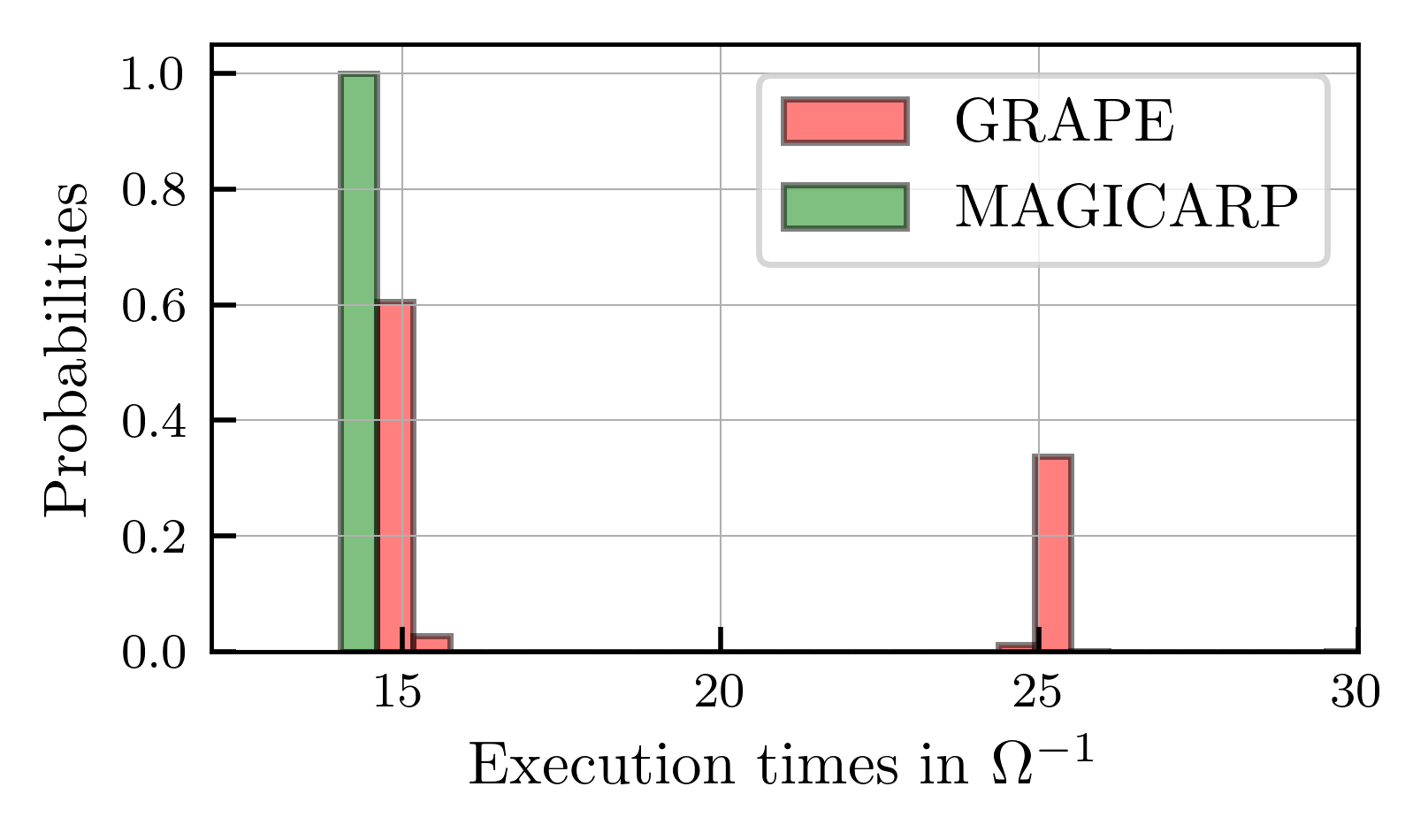}%
\label{subfig:GRAPE-vs-MAGICARP-X}%
}
\hfill
\subfloat[]{
\includegraphics[width=0.45\textwidth]{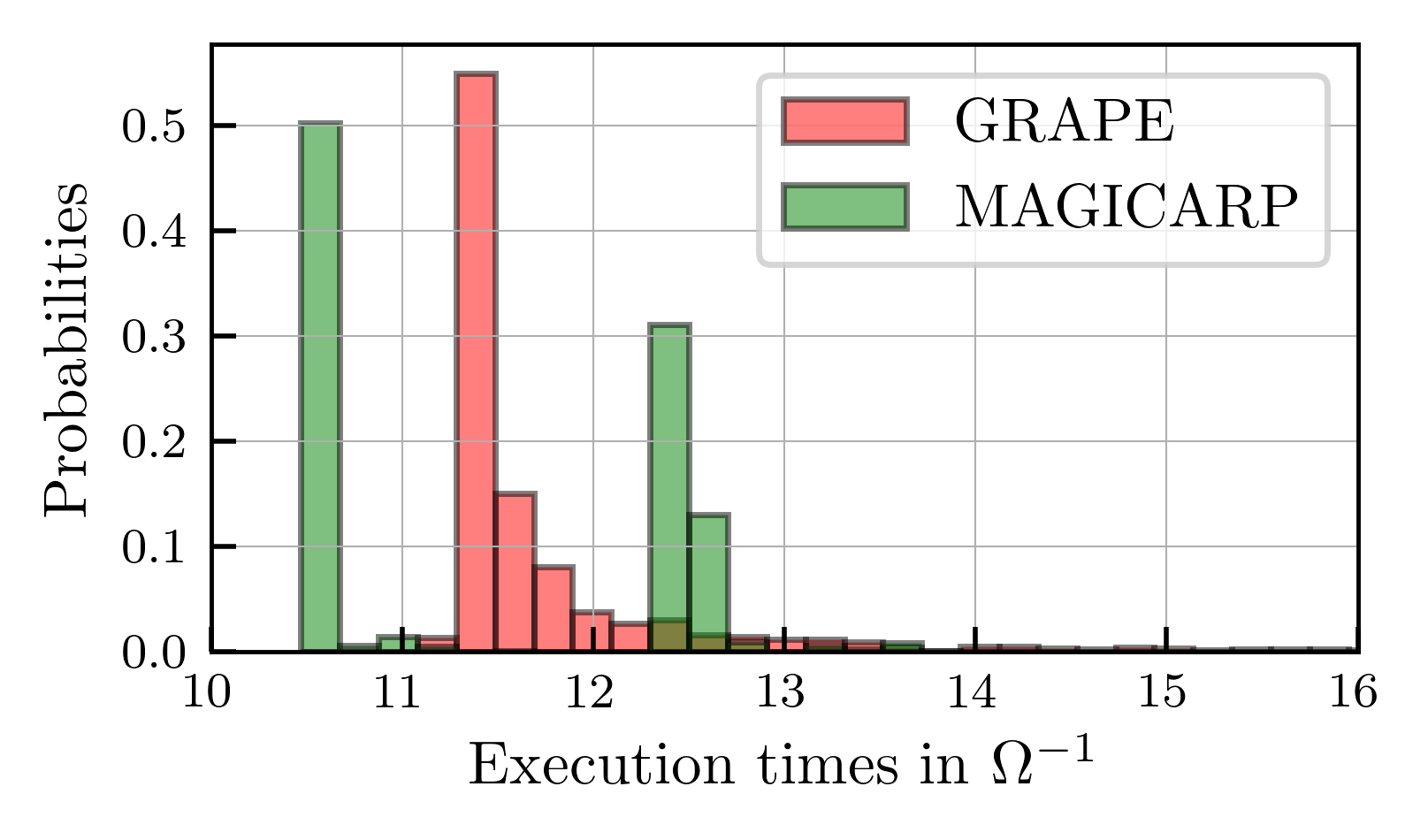}%
\label{subfig:GRAPE-vs-MAGICARP-SUMX}%
}
\caption{Comparison of GRAPE and MAGICARP execution times for 1000 random initializations on a $\text{Tb}_2\text{Pc}_3$ molecule.
(a) QFT gate.
(b) T gate.
(c) X gate.
(d) SUMX gate.}\label{fig:GRAPE-vs-GRD-specific}
\end{figure}

\section{Conclusion and perspectives}

This work presents a numerical proof of concept for the design of smooth quantum control pulses using a parametrization derived from the Pontryagin maximum principle. The resulting method, termed MAGICARP, is applied to the implementation of quantum gates in qudit systems while neglecting the drift Hamiltonian. Compared with the standard GRAPE approach, MAGICARP makes use of the geometry of the unitary group together with the system transition graph to achieve substantially shorter execution times for qudit systems of dimension up to $d = 16$. This advantage becomes increasingly pronounced as the system dimension grows.

The method is demonstrated on the $d=4$ and $d=16$  single-molecule magnets $\text{TbPc}_2$ and $\text{Tb}_2\text{Pc}_3$, respectively, while explicitly accounting for several experimentally relevant constraints specific to these platforms. Indeed, to implement uniformly sampled unitary gates  on the triple-decker molecule, the median execution time found by MAGICARP is about 20\% faster relative to the one found by GRAPE. Similarly, in the case of the QFT gate run on the same system, the minimum gate time obtained with the MAGICARP algorithm  was roughly 74\% of that from the GRAPE method.

Several promising directions for future research deserve further investigation. A first natural step would be to experimentally validate the effectiveness of the control pulses derived from MAGICARP.  Our method may provide an efficient quantum time-optimal control framework for a wide range of finite-dimensional controllable quantum systems~\cite{Quantum_optimal_control_review2}, including NV centers~\cite{diamond_centers}, NMR platforms~\cite{NMR}, trapped ions~\cite{trapped_ions}, and Rydberg atoms. Another important extension would be to generalize the MAGICARP parametrization framework to closed quantum systems by accounting explicitly for the drift Hamiltonian, as well as to open quantum systems~\cite{open_loop_review}. This could be achieved either by explicitly incorporating decoherence effects into the cost functional~\cite{PhysRevA.77.032117,Hartmann2025nonlinearityof}, or by including them directly in the numerical evolution~\cite{PhysRevA.99.052327,Le2026}.

Further progress may also come from identifying optimization and parametrization strategies capable of enhancing MAGICARP. Promising directions include the development of more efficient gradient-based optimization algorithms, the use of multiple-shooting techniques~\cite{A_time-parallel_multiple-shooting_method}, and the implementation of scalable high-dimensional quantum trajectory solvers~\cite{doi:10.1137/24M1690795}. It will also be important to clarify the precise regime in which MAGICARP retains its advantage, particularly in the presence of model mismatch and experimental imperfections such as control-amplitude errors, frequency detuning, and uncertainties in the Hamiltonian parameters~\cite{instrumental_distortions,Fisher2011OptimalControl_and_instrumental_distortions}.

Another promising avenue for future work is to examine how the connectivity pattern of the graph corresponding to the energy levels of different physical systems affects the method’s efficiency, particularly in terms of execution time. More broadly, future work could examine whether the MAGICARP parametrization reveals interpretable geometric features of efficient quantum control mechanisms, such as simpler trajectories in unitary space or recurring structures in near-time-optimal pulse shapes.


\section*{Acknowledgements} The authors would like to express their gratitude to Benjamin Bakri for their fruitful discussions.
The authors would like to acknowledge the High Performance Computing Center of the University of Strasbourg for supporting this work by providing scientific support and access to computing resources. Part of the computing resources were funded by the Equipex Equip@Meso project (Programme Investissements d'Avenir) and the CPER Alsacalcul/Big Data.
This work of the Interdisciplinary Thematic Institute QMat, as part of the ITI 2021-2028 program of the University of Strasbourg, CNRS and Inserm, was supported by IdEx Unistra (ANR-10-IDEX-0002), and by SFRI-STRAT’US project (ANR 20 SIFRI 0012) and EUR QMAT (QMAT ANR-17-EURE-0024) under the framework of the French Investments for the Future Program. During the writing and submission process, D.J. acknowledges the additional financial support from the Institute for Basic Science, Republic of Korea (IBS-R027-D1) and J-G.H. acknowledges the additional financial support of the CNRS UMR-7177 and Fondation Jean-Marie Lehn. This work was carried out as part of the Inria Exploratory Action MARCQ, supported by Groupe La Poste, a patron of the Inria Foundation.

\newpage
\appendix
\section{MAGICARP}  
\subsection{Importance of the coupling graph}
In this work, the numerical experiments were done assuming either the linearly coupled graph of the double decker in~\cref{fig:double_levels} or the coupling graph of the triple decker in~\cref{fig:graph}. 
Since the coupling graph depends on the choice of physical molecules and platforms, the execution times for a given gate are likely to differ and to depend on the Hilbert space dimension. 

The simplest configuration is the linear graph in which each level is only coupled to its nearest neighbors. At the other end of the spectrum, there is the fully connected graph in which each level is coupled to every other levels. We can assess how the GRAPE and MAGICARP methods take advantage of the growing connectivity of the graph by comparing in~\cref{fig:connectivity} each method with itself in two situations: when the graph is fully connected and when the the graph is linearly coupled. 

In~\cref{subfig:connectivity-QFT}, the ratios of execution times are above 1 which means that both GRAPE and MAGICARP consistently benefit from a fully connected graph. This is reasonable because, the more edges in the graph, the more available controls. However this improvement appears to be all the more significant as the Hilbert space dimension increases. 

For gates sampled uniformly at random in~\cref{subfig:connectivity-random}, MAGICARP slightly benefits from the use of a complete graph instead of a linear one (for every dimension, gates are in median time around 20\% faster) but not as much as GRAPE does. This suggest that MAGICARP is able to better exploit sparsely connected graphs.
\begin{figure}[hbt]
    \centering
    \subfloat[\label{subfig:connectivity-random}]{%
        \includegraphics[width=0.45\textwidth]{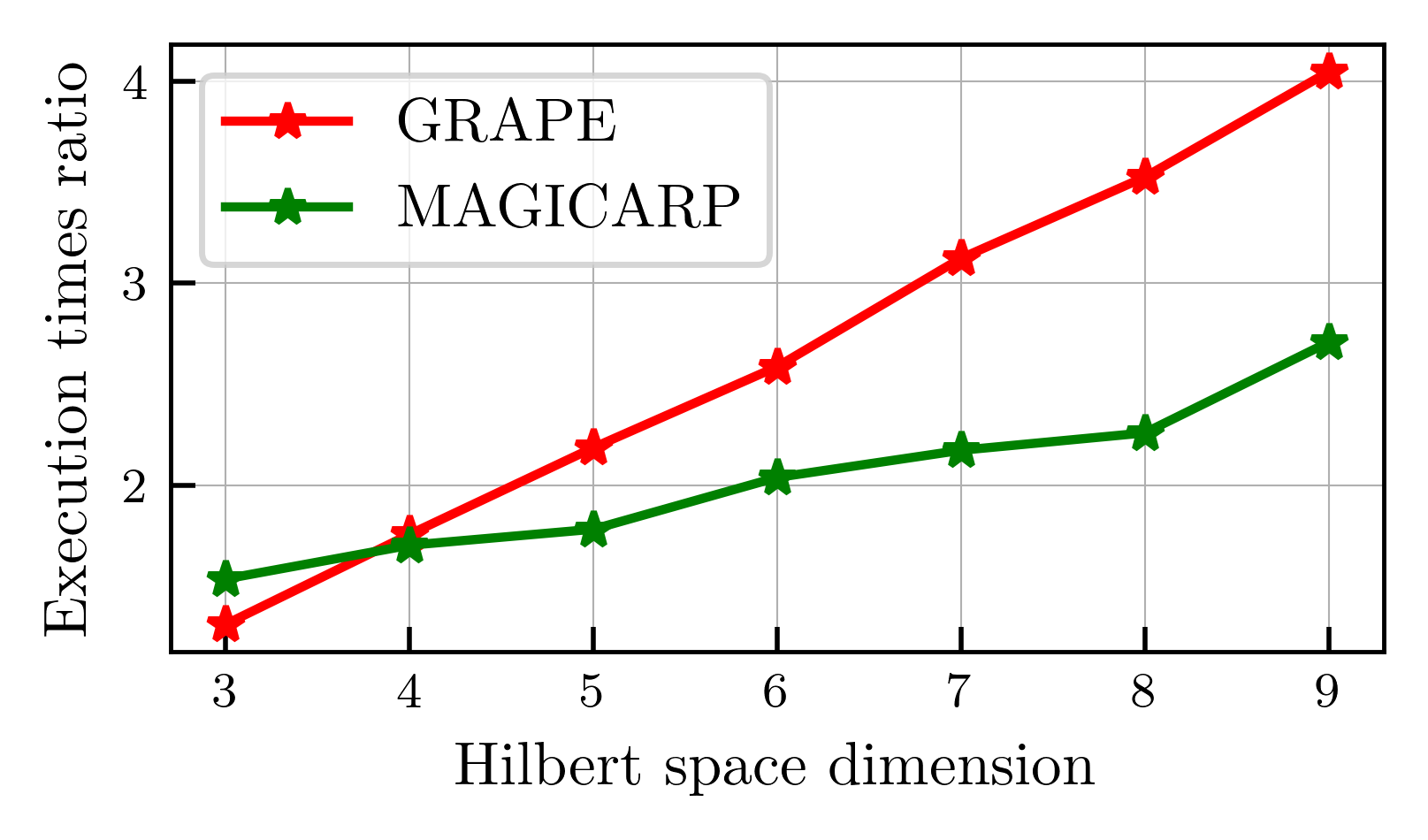}%
    }
    \hfill
    \subfloat[\label{subfig:connectivity-QFT}]{%
        \includegraphics[width=0.45\textwidth]{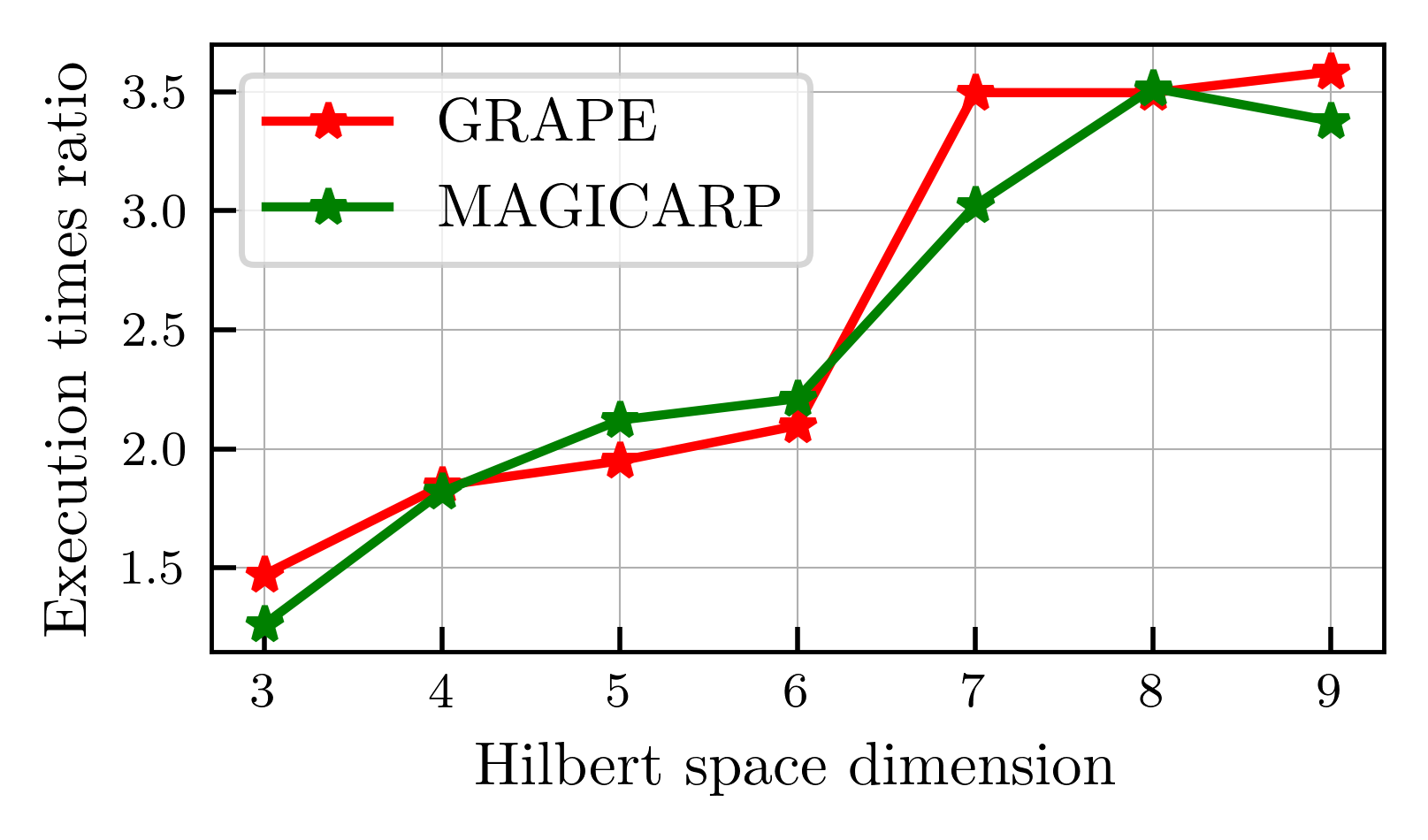}%
    }
    \caption{Comparison of the median (a) or the minimum time (b) to implement at least a $10^2$ gates with infidelity below $10^{-4}$. For a given gate, each of the two methods is compared to itself assuming either a complete graph or a linear graph. The graphs have from 3 to 9 nodes depending on the Hilbert space dimension but given a dimension, the connectedness varies. For each method, GRAPE or MAGICARP, the quantity shown is the ratio between the execution time on a linear graph and the runtime on a complete graph, effectively describing how each method is faster when being run on a complete graph instead of a linear one.\ (a) Ratio of the median execution times for $10^{2}$ random Haar gates.\ (b) Ratio of the minimum execution times for the QFT gate.}\label{fig:connectivity}
\end{figure}

\subsection{Importance of a \texorpdfstring{$\sigma_z$}{sigma\_z} control}
In this work, the numerical simulations were done assuming we have access to a control on $\sigma_x$ and $\sigma_y$ for each edge between two levels of the Hilbert space, similar to~\cref{eq:godfrin}, but no control on $\sigma_z$ because for experimental reasons this direction is not available on some physical platforms, especially those based on SMMs.

To assess how much of an improvement we can hope for by lifting this experimental constraint on the control directions, we compare how GRAPE and MAGICARP take advantage of having a control $\sigma_z$ depending on the coupling graph. This is what we test for in~\cref{fig:sigmaz_available}.

For each method and almost all coupling graphs, we observe that there is little improvement in both the median time of~\cref{subfig:sigmaz_available-random} and the minimum time of~\cref{subfig:sigmaz_available-QFT} for the QFT gate. However, on a fully connected graph and for random target gates, MAGICARP appears to benefit significantly from a $\sigma_z$ control, especially as the dimension increases.
\begin{figure}[hbt]
    \centering
    \subfloat[\label{subfig:sigmaz_available-random}]{%
        \includegraphics[width=0.45\textwidth]{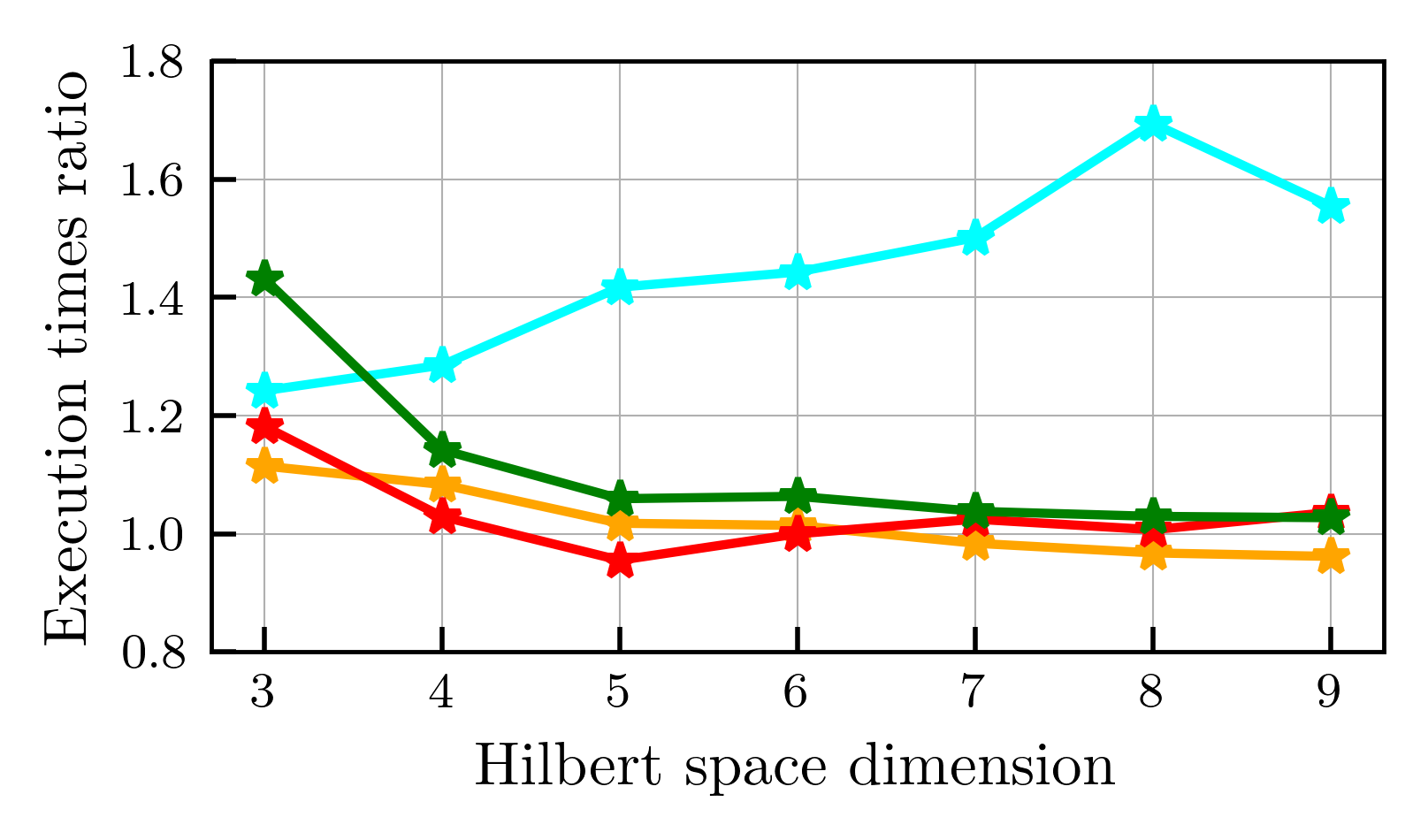}%
    }
    \hfill
    \subfloat[\label{subfig:sigmaz_available-QFT}]{%
        \includegraphics[width=0.45\textwidth]{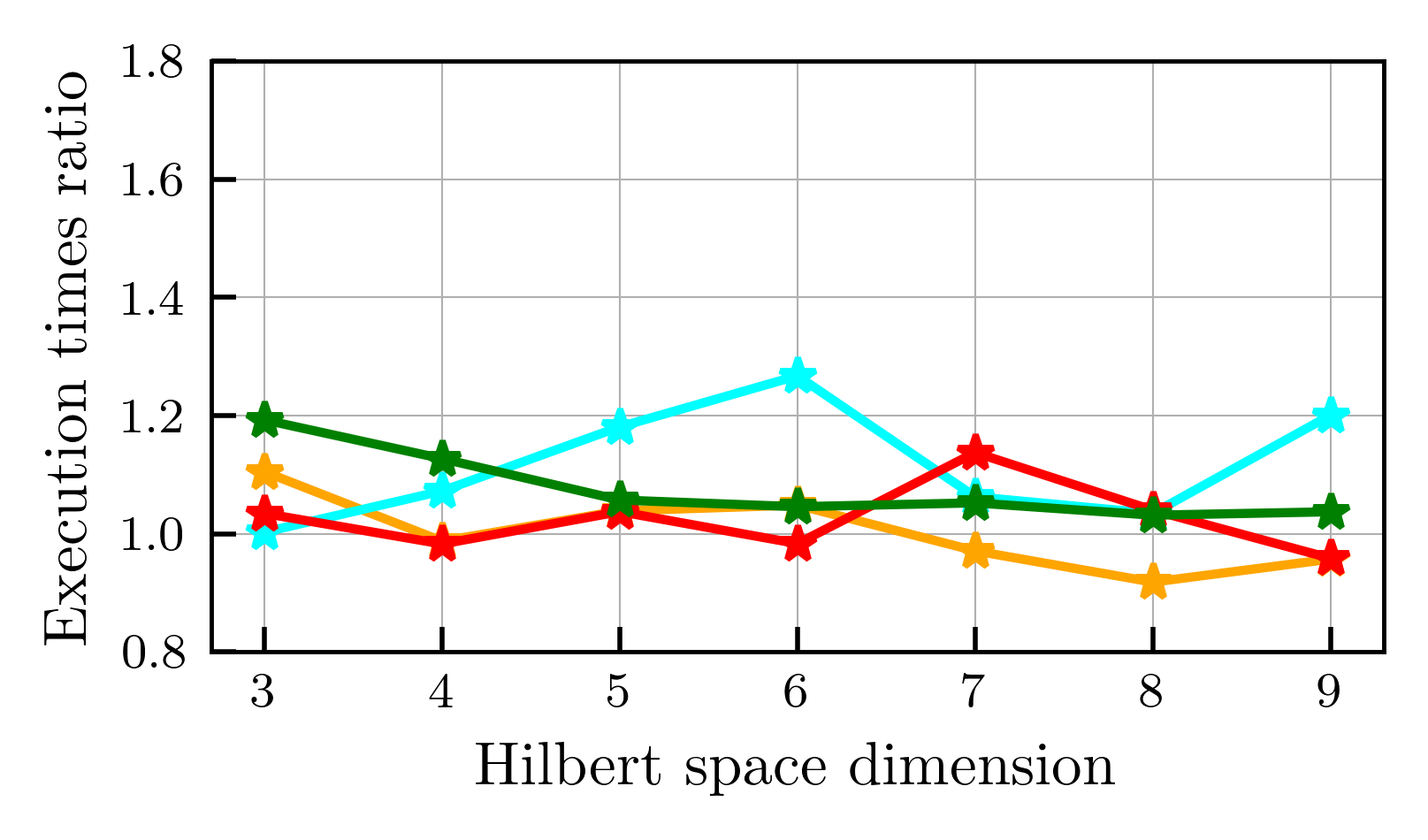}%
    }
    \caption{
    Comparison of the median (a) or the minimum time (b) to implement at least $10^{2}$ gates with infidelity below $10^{-4}$ using MAGICARP and GRAPE. The coupling graph is either linear or fully connected and its number of nodes is the dimension of the Hilbert space.
    For each method, the quantity shown is the ratio between the execution time with controls $(\sigma_x, \sigma_y)$ and and the execution time with controls $(\sigma_x, \sigma_y, \sigma_z)$. This quantity describes how much an additional control direction lowers the execution time depending on the coupling graph.
    Red: GRAPE linear graph. Orange: GRAPE fully connected graph. Green: MAGICARP linear graph. Cyan: MAGICARP fully connected graph.
    (a) Ratio of the median execution times for $10^{2}$ random Haar gates. (b) Ratio of the minimum execution times for the QFT gate.
    }\label{fig:sigmaz_available}
\end{figure}

\label{app:Benchmark}
\section{Benchmark}
\subsection{Improvement of GRAPE with the target infidelity}\label{sec:1e-4_good_choice_GRAPE}
The execution times found by GRAPE depend on the target infidelity threshold, as illustrated by~\cref{fig:GRAPE_fid}. In general, the higher this threshold, the lower the execution time because a high target infidelity means fewer constraints on the control pulses.
In this section, we discuss why we chose to compare execution times at $10^{-4}$ target infidelity.
\begin{figure}[hbt]
    \centering
    \subfloat[]{%
        \includegraphics[width=0.45\textwidth]{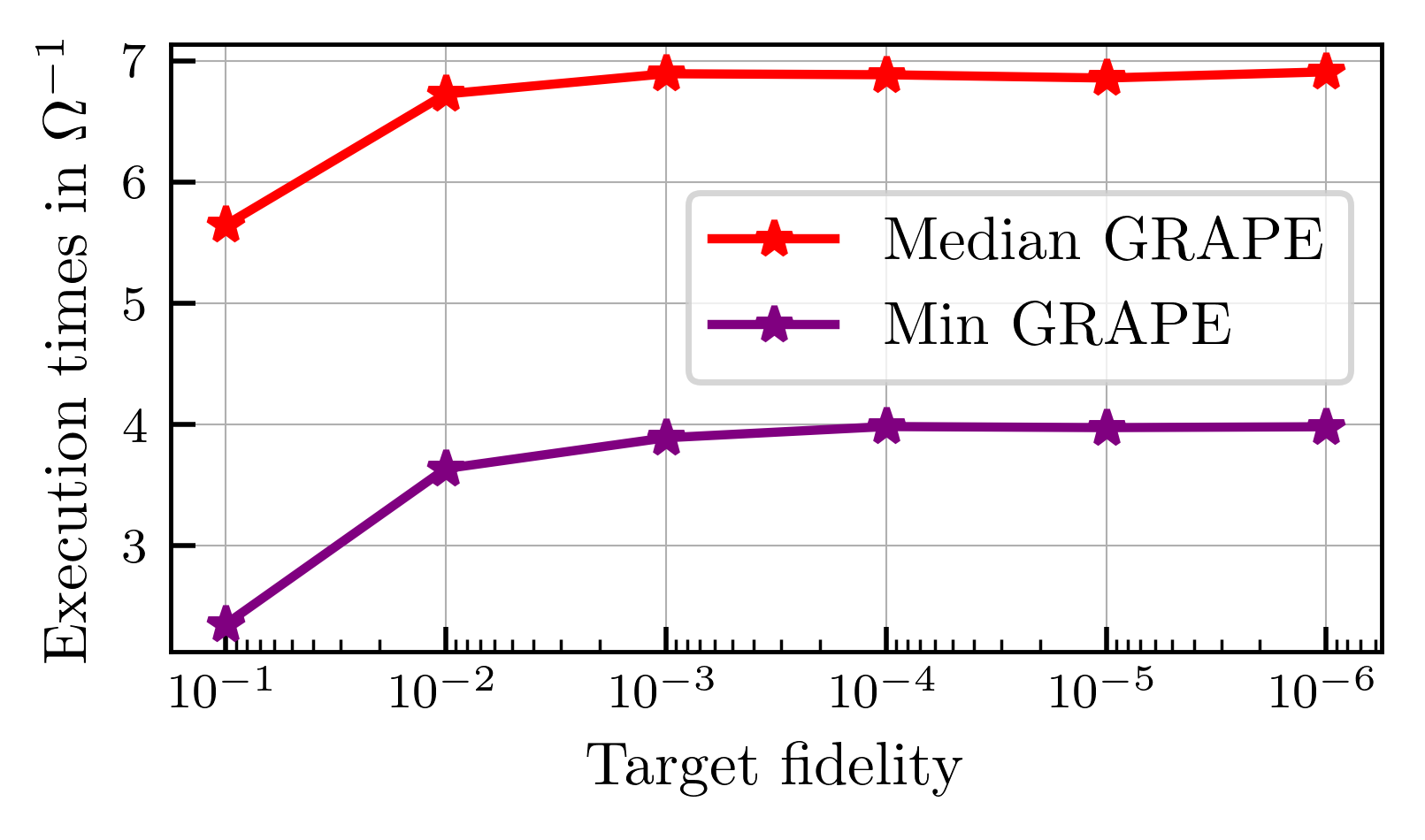}%
    }
    \hfill
    \subfloat[]{%
        \includegraphics[width=0.45\textwidth]{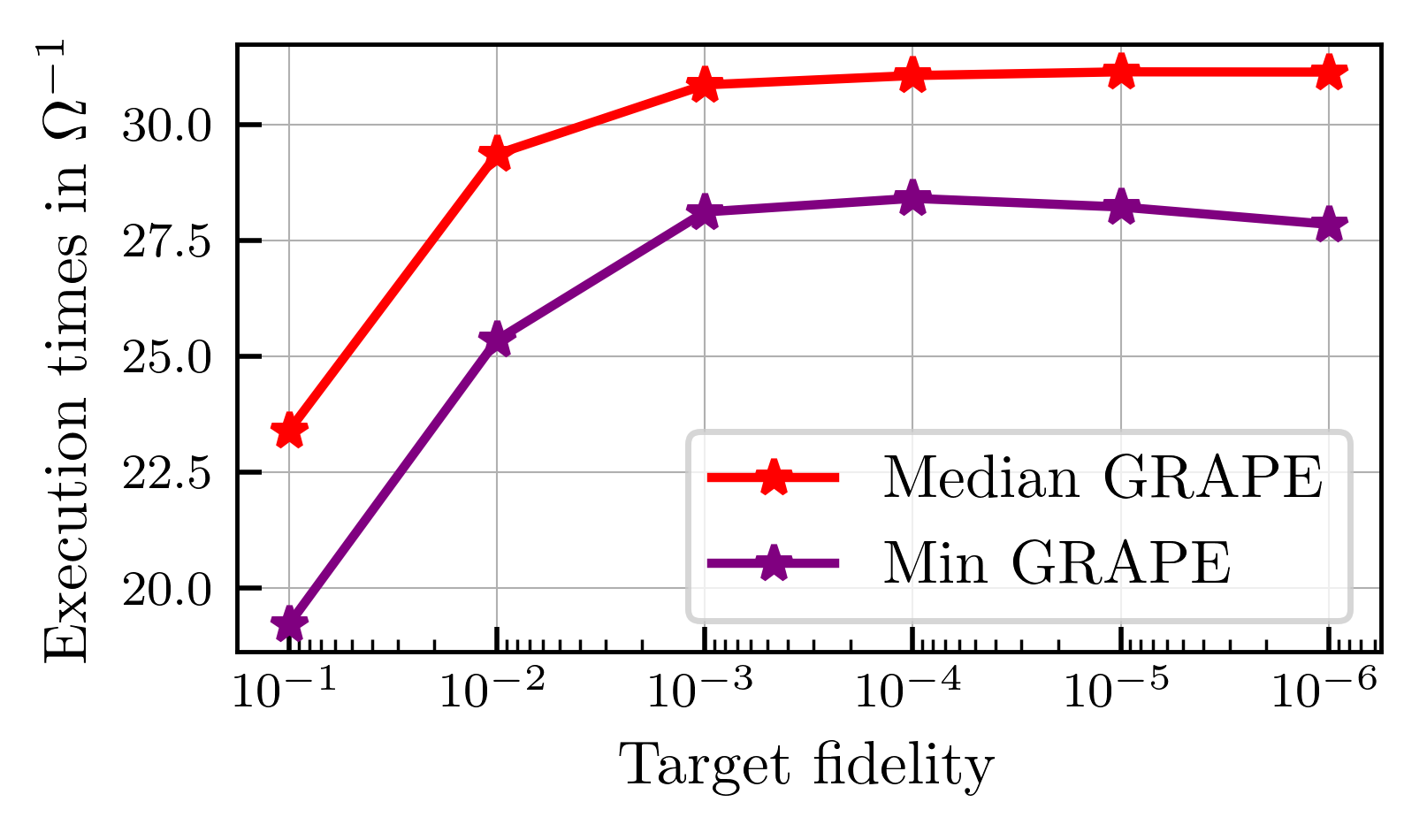}%
    }
    \caption{Comparison between the median (a) and minimum time (b) found by GRAPE to implement the \emph{same} $10^{3}$ random Haar gates with different target infidelity levels. (a) Double decker (linearly coupled) graph with a Hilbert space dimension $d = 4$. (b) Triple decker graph with a Hilbert space dimension $d = 16$.}\label{fig:GRAPE_fid}
\end{figure}
First, this level of accuracy allows for shorter execution times while corresponding in several platforms to the order of magnitude of the errors induced by the noise~\cite{1e-4_fid_is_good}. Second, for high-dimensional systems, finding more accurate pulses in the MAGICARP framework proved difficult because it required significantly more iterations of the optimization algorithm we used. Third,~\cref{fig:GRAPE_fid} suggests that the execution times found by GRAPE at target infidelities $\leq 10^{-4}$ are not signficicantly slower nor faster than the execution time found when the target infidelity is set to $10^{-4}$. 

\subsection{Improvement of GRAPE with the number of control pieces}
The GRAPE method relies on piecewise constant approximations of control pulses. This means that the more pieces are available, the more accurate the approximation can be. In particular, to some extent, the execution time found by GRAPE shortens as the number of pieces increases. In this work, we have chosen the number of pieces per control pulse to be $10(d + 1)$ where $d$ is the dimension of the Hilbert space because we found empirically that choosing a higher number of pieces only allows for negligible improvements of the execution time, as illustrated by~\cref{fig:GRAPE_n_ts}. 

\begin{figure}[hbt]
    \centering
    \subfloat[]{%
        \includegraphics[width=0.45\textwidth]{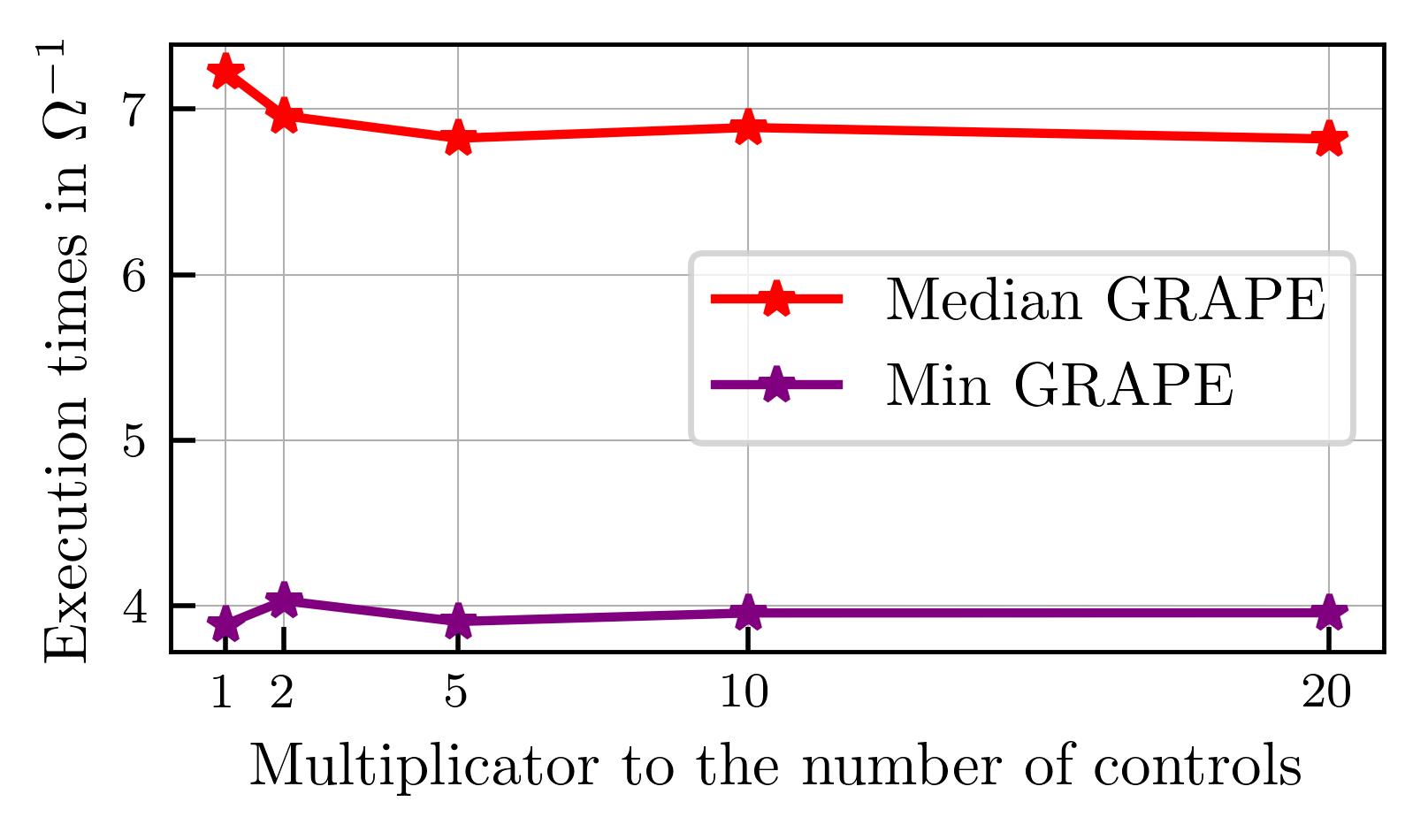}%
    }
    \hfill
    \subfloat[]{%
        \includegraphics[width=0.45\textwidth]{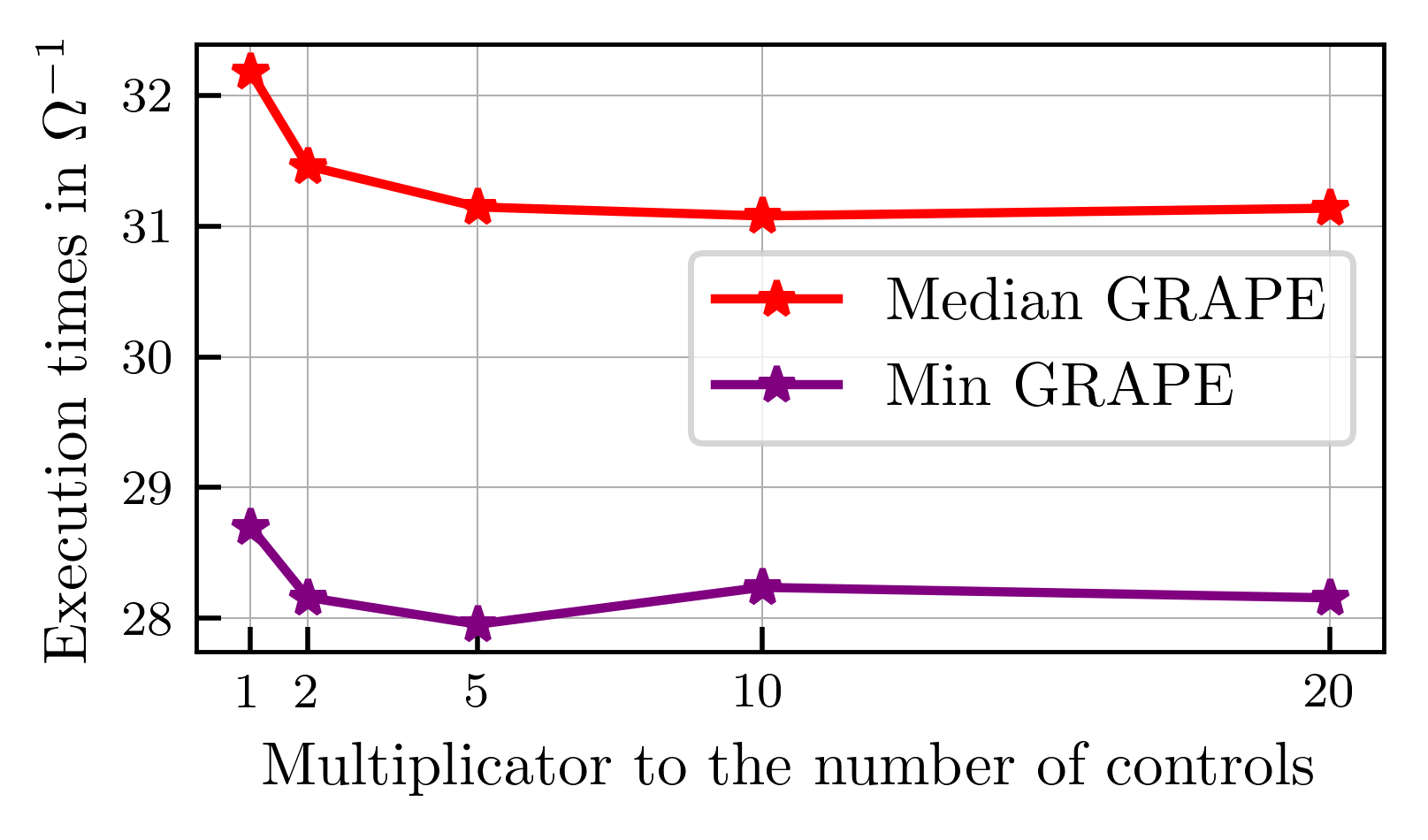}%
    }
    \caption{
    Comparison between the median and minimum time (b) to implement the \emph{same} $10^{3}$ random Haar gates using different number of pieces for the piecewise constant pulses optimized by the GRAPE method. The $x$-axis reflects the number of pieces which, given the dimension $d$ of the Hilbert space, is set to $x(d + 1)$. The target infidelity is $10^{4}$. (a) Double decker (linearly coupled) graph with a Hilbert space dimension $d = 4$. (b) Triple decker graph with a Hilbert space dimension $d = 16$.}
    \label{fig:GRAPE_n_ts}
\end{figure}
\section{Numerical framework}~\label{app:Numerical_framework}

\paragraph{Objectives.} Our objectives in this section is twofold. 
We provide closed-form expressions for the derivative of the map that associates to every control parameter $M$ the final propagator of Schrödinger's equation. We outline its calculation which relies on the so-called adjoint state method. This provides the tools for implementing any first-order optimization routine and to adapt the $\z$ parametrization of controls to a different problem. On the other hand, we discuss our implementation of the natural gradient descent, an optimization routine used to produce some of the figures in this paper. The corresponding Julia and Python codes are available on the Github repositories \href{https://github.com/killianlutz/pyMagicarp}{\texttt{Magicarp}} and \href{https://github.com/killianlutz/pyMagicarp}{\texttt{pyMagicarp}}.

\subsection{Two equivalent optimal control problems in the driftless case}
We denote the set of traceless Hermitian matrices of order $d$ by $\iu \mathfrak{su}(d)$.
\begin{assumption}[Target]
    $$\target \in SU(d), \quad \target \neq I$$
\end{assumption}
\begin{assumption}[Orthonormal control Hamiltonians]\label{assumption:orthonormal}
    \begin{equation*}
        H_j \in \iu \mathfrak{su}(d), \quad \RE\tr(H_iH_j) = \delta_{i,j}, \quad 1 \leq i,j \leq m
    \end{equation*}
\end{assumption}
\begin{assumption}[Lie algebra rank condition]\label{assumption:LARC}
    $$\operatorname{Lie}(\iu H_1, \ldots, \iu H_m) = \mathfrak{su}(d).$$
\end{assumption}
Let us comment on the second and third assumptions. In some applications the control Hamiltonians are related to the Pauli matrices which are orthogonal as in Assumption~\ref{assumption:orthonormal}. As such this assumption is not too restrictive. Together with the normalization condition, it is convenient but by no means necessary.

The third Assumption~\ref{assumption:LARC} is common in the practical coherent control of quantum systems because the control Hamiltonians often do not linearly span the set of tangent vectors to the special unitary group at the identity matrix. The Lie algebra rank condition is about iterated commutators of the control Hamiltonians. When it holds, combining the available control directions provides a family of tangent vectors that is sufficiently rich to move in any admissible direction. Intuitively this is because to second order in $t$ close to zero, $e^{-tA}e^{-tB}e^{tA}e^{tB} = I + t^2[A, B] + o(t^2)$. All in all Assumption~\ref{assumption:LARC} is meant to guarantee that any desired special unitary gate is reachable by means of an appropriate sequence of constant pulses.

\paragraph{Two optimal control problems.}
Let $P$ stand for the orthogonal projector onto the real linear span of $H_1, \ldots, H_m$ for the Frobenius (Hilbert-Schmidt) norm $|M| = \tr(MM^\dagger)^{1/2}$. In particular using Assumption~\ref{assumption:orthonormal}, 
\begin{equation}
    |P(M)|^2 = \sum_{j=1}^m \operatorname{Re}\tr(H_jM)^2.
\end{equation}
We denote by $|u|$ the euclidean norm of a vector $u \in \RR^m$. Let $\Omega_0$ be an amplitude fixed to 1 and $\Omega$ a total amplitude constraint, with $[\Omega] = [\Omega_0] = T^{-1}$. Consider the following two optimal control problems. 

\noindent\textbf{Problem A: free final time with amplitude constraint}.
\begin{align}\tag{$\mathscr{P}_u$}\label{ocp:u}
    \inf_{u \in L^\infty(0, T; \RR^m)} T 
\end{align}
subject to the pointwise amplitude constraint $|u(t)|\leq \Omega/\Omega_0$ and to the terminal constraint $\x(T) = \x_f$, where the propagator $\x$ solves the Cauchy problem 
\begin{equation}
    \x(0) = I,  \quad \dot{\x}(t) = -\iu \Omega_0 \sum_{j=1}^{m}u_j(t)H_j U(t), \quad 0 < t < T.
\end{equation}

\noindent\textbf{Problem B: initial covector over a fixed time horizon}.
\begin{align}\tag{$\mathscr{P}_{\z}$}\label{ocp:z}
    \inf_{\z \in \iu \mathfrak{su}(d)} \hnorm{P(\z)}
\end{align}
subject to the terminal constraint $\x(1) = \x_f$, where the propagator $\x$ solves the Cauchy problem\begin{equation}\label{eq:schrod}\x(0) = I, \quad \dot{\x}(s) = -\iu \Omega_0 \sum_{j=1}^{m}\RE \tr\left(\x(s)^\dagger H_j \x(s)\z\right)H_j U(s), \quad 0 < s < 1.
\end{equation}

Using Assumption~\ref{assumption:LARC} and rescaling, we can build an admissible control for Problem~\eqref{ocp:u} which steers the identity to $\target$ while satisfying the experimental constraint $|u(t)| \leq \Omega/\Omega_0$. This relies on the Chow-Rashevskii theorem or more specifically~\cite[Corollary 3.1]{albertini2001notions}. Moreover it can be shown that an optimal solution exists and that, using the maximization condition from Pontryagin's maximum principle (see e.g.~\cite{boscain2021introduction}), any \emph{normal} extremal is admissible for Problem~\eqref{ocp:z} with the physical time horizon $T$ and control $u$ given by
\begin{equation}\label{eq:equivalence}
        \Omega T = \hnorm{P(\z)}, \quad \text{ and } \quad u_j(t) =\frac{\Omega}{\Omega_0} \RE \tr \left( \x(s)^\dagger H_j \x(s)\frac{\z}{\hnorm{P(\z)}} \right), \quad \text{where} \quad \quad Ts = t.
\end{equation}
The time $T$ in the first problem corresponds in the second problem to the norm of the projection of $\z$ onto the space of controllable directions relative to the available control effort, that is $|P(\z)|/\Omega$. 

\begin{remark}[Restriction to normal extremals]\label{rem:abnormal_extremals}
    Not all Pontryagin extremal are normal and hence, it is \emph{not} true in general that every optimal control $u$ of Problem~\eqref{ocp:u} is of the form given by Equation~\eqref{eq:equivalence}. Depending on the choice of control Hamiltonians, one can prove that any abnormal extremal, which is not normal at the same time, cannot be optimal and hence can be disregarded. This occurs for instance in the $K+P$ setting. The proof given in~\cite[Appendix C]{boscain2002k+} relies on the Cartan decomposition of Lie groups and the Goh necessary conditions derived in~\cite{agrachev1996abnormal}.
\end{remark}

In Problem~\eqref{ocp:u} the constraint $|u(t)| \leq \Omega/\Omega_0$ physically reflects an upper bound on the power transmitted to the qudit and, mathematically, such a constraint ensures the existence of an optimal control. Indeed without this constraint and if $\x_f$ is non-trivial, given a control $u$ driven the system $\x$ from $\x(0) = I$ to $\x(T) = \x_f$ in time $T$, the more powerful control $u_n(t) = n u(t)$ performs the transfer in the shorter time $T_n = T/n$. Since $n$ can be arbitrarily large, the transfer can be performed in arbitrarily short time. However there is no control performing the transfer instantaneously with $T = 0$ because this would imply $\x(T) = \x(0)$ but $\x_f \neq I$ by assumption. In other words, the infimum is zero but is not attained and therefore not a minimum.

\begin{remark}[Notations]
    In the sequel, we work in units such that $\Omega_0 = \Omega = 1$ without loss of generality.
\end{remark}
\subsection{Differential of the end-point mapping and its adjoint operator} \label{sub:computed_gradient}
In the framework of Problem~\eqref{ocp:z}, the state $\x : [0, 1] \to SU(d)$ solves a non-linear equation with feedback control determined by $\z \in \iu \mathfrak{su}(d)$. More precisely we are interested in the Cauchy problem
\begin{align}\label{eq:cauchy}
    \x(0) = I, \quad \dot{\x} & = -\iu \sum_{j=1}^{m} \RE \tr(\x^\dagger H_j \x\z)H_j\x, \quad t \in (0, 1).
\end{align}
Let $E : \iu \mathfrak{su}(d) \to SU(d)$ denote the end-point mapping, defined by $E(\z) = \x(1)$ where $\x$ is the solution of~\eqref{eq:cauchy} with parameter $\z$. 
The derivative $E'$ of this mapping is needed to implement a first order optimization routine in the so-called optimize-then-discretize framework. We now focus on computing this derivative.

Equip the space of traceless Hermitian matrices $\iu \mathfrak{su}(d)$ and $\CC^{d \times d}$ with the inner-product $\langle a, b\rangle = \RE\tr(a b^\dagger)$. 
Then at any point $\z$, the derivative $E'(\z) : \iu \mathfrak{su}(d) \to \CC^{d \times d}$ and its adjoint $E'(\z)^\dagger : \CC^{d \times d} \to \iu \mathfrak{su}(d)$ are obtained as follows. As a shorthand notation, we write $h_j(t) = \x(t)^\dagger H_j \x(t)$.

\begin{lemma}[Jacobian-vector product]
    \label{lem:jvp}
    Given a perturbation $\delta \z \in \iu \mathfrak{su}(d)$, the resulting perturbation $\delta \x(1)$ of the end-point $\x(1)$ of the curve $\x(t)$ is given by $E'(\z)\delta \z = -\iu \x(1)w(1)$ where $w(t) \in \iu \mathfrak{su}(d)$ satisfies
    \begin{align}\label{eq:linearization}
        w(0) = 0, \quad
        \dot{w}(t) = \sum_{j=1}^{m} \RE \tr \left( h_j(t) \delta \z - 2 \iu \z h_j(t) w(t) \right)h_j(t), \quad 0 < t < 1
    \end{align}
\end{lemma}

\begin{lemma}[Vector-jacobian product] \label{lem:vjp}
    Given a linear form $S \in \CC^{d \times d}$, the corresponding form $R = E'(\z)^\dagger S$ in $\iu \mathfrak{su}(d)$ is given by
    \begin{align}\label{eq:grad}
        R = \sum_{j=1}^{m} \int_{0}^{1}\RE \tr (h_j(t) v(t)^\dagger)h_j(t) \d t
    \end{align}
    where $v(t) \in \CC^{d \times d}$ satisfies
    \begin{align}\label{eq:adjoint_state}
        v(1) = \iu \x(1)^\dagger S, \quad \dot{v}(t) = -2 \iu\sum_{j=1}^{m}\RE \tr (h_j(t) v(t)^\dagger) h_j(t) \z, \quad 0 < t < 1
    \end{align}
\end{lemma}
We postpone the proofs of these results to Section~\ref{sec:proof_jvp_vjp}. The calculation relies on the so-called adjoint-state method which we briefly outline and illustrate on a concrete example.

\subsubsection{The adjoint state method}~\label{app:adjoint_state_method}
 Calculating the cost function $J(\z)$ involves the implicit solution $\x(t)$, to an equation of the form $F(\x(\cdot), \z) = 0$ with parameter $\z$. Loosely speaking implicit means that finding the solution $\x$ given $\z$ is a challenging computational task for which convenient closed form expressions cannot be exploited. In particular finding the gradient of $J$ is challenging because it is not given directly by a formula. The method is often a computationally efficient way to overcome this issue. At the core of the adjoint-state method lies the implicit function theorem and duality.

\paragraph{Method outline.}
Suppose that the vector parameter $\z$ has $n$ real coordinates and that the cost function is $$J(\z) = \ell(x(\z))$$ where $x(\z) \in \RR^N$ is defined as the unique solution to an implicit equation $F(x(\z), g) = 0$.
Then in response to a perturbation $\delta \z$ of the parameter $\z$, the first variation of $J$ is
\begin{equation}\label{eq:intermediate}\delta J = \left.\frac{\mathrm{d}}{\mathrm{d}\varepsilon} J(\z + \varepsilon \delta \z)\right|_{\varepsilon = 0} = \langle \nabla \ell(x(\z)), \delta x \rangle\end{equation}
where $\delta x$ is the first variation of $x(\z)$ given by
\begin{equation}\delta x = \left.\frac{\mathrm{d}}{\mathrm{d}\varepsilon} x(\z + \varepsilon \delta \z)\right|_{\varepsilon = 0}\end{equation}
We are evaluating a linear form on the implicit solution to a linear equation depending on $\delta \z$.
To relate $\delta x$ to $\delta \z$ we match the first order variations in the equation $F(x(\z), \z) = 0$. This implies that $\delta x$ solves a linear system of equation 
\begin{equation}\label{eq:adjoint_eq}A \delta x + B \delta \z = 0\end{equation}
where $A = \partial_x F(x(\z), \z)$ and $B = \partial_{\z} F(x(\z), \z)$ are the jacobians of $F$ with respect to each arguments. 
The goal is to rewrite Equation~\eqref{eq:intermediate} in the form $\delta J = \langle h, \delta \z \rangle$ where the vector $h \in \RR^n$ will be the gradient we seek. This is what enables the introduction of an adjoint state $v \in \RR^N$ solution to the adjoint linear system 
$$A^\dagger v = \nabla \ell(x(\z))$$
Because using Equation~\eqref{eq:adjoint_eq} the first variation of $J$ rewrites 
$$\delta J = \langle x(\z), \delta x \rangle = \langle A^\dagger v, \delta x \rangle = \langle v, A\delta x \rangle = \langle v, -B\delta \z \rangle$$
which means that the gradient of $J$ is the vector $$\nabla J(\z) = -B^\dagger v$$ which depends on $\z$ via $B$ and the adjoint state $v$.

\paragraph{Concrete example.} For example put $n = N = 2$. Define $\ell(x) = |x|^2/2$, the squared euclidean norm of $x$ and $$F(x, \z) = \begin{pmatrix}\arctan(x_1) - \z_1 \\ x_2 - (x_1 + \z_2)^2\end{pmatrix}.$$ If $|\z_1| < \pi/2$ then the vector $x(\z)$ solution to $F(x, \z) = 0$ is given by \begin{equation*}
    x(\z) = \begin{pmatrix}
		\tan \z_1 \\ (\tan \z_1 + \z_2)^2
	\end{pmatrix} 
\end{equation*} 
This means that $J(\z) = \tan(\z_1)^2/2 + (\tan(\z_1) + \z_2)^4/2$ and so
\begin{align}\label{eq:gradient_usual}
    \partial_{g_1} J(\z) &= 
    (1 + \tan(\z_1)^2)\left(\tan(\z_1) + 2(\tan(\z_1) + \z_2)^3\right), \\ 
	\partial_{g_2} J(\z) &= 2(\tan(\z_1) + \z_2)^3.
\end{align}
Let us obtain the same expression for $\nabla J(\z)$ using the adjoint state method described above. With the shorthand notation $y = x(\z)$, the coordinates of $y$ are $y_1 = \tan \z_1$ and $y_2 = \tan \z_1 + \z_2$ and so we have
\begin{equation*}
    \nabla \ell (y) = y, \quad A = \begin{pmatrix}
	(1 + y_1^2)^{-1} & 0 \\ -2(y_1 + \z_2) & 1
\end{pmatrix}, \quad B = \begin{pmatrix}
	-1 & 0 \\ 0 & -2(y_1 + \z_2)
\end{pmatrix}
\end{equation*}
Therefore the adjoint state $v$ has coordinates $v_1 =(y_1 + 2y_1y_2 + 2\z_2y_2)(1 + y_1^2)$ and $v_2 = y_2$. Accordingly the gradient $\nabla J(\z)$ should be given by
\begin{equation}\label{eq:adjoint_gradient}
    \nabla J(\z) = \begin{pmatrix}
    v_1 \\
	2(y_1 + \z_2)v_2
\end{pmatrix}
\end{equation}
and it is straightforward to verify that the vector in Eq.~\eqref{eq:adjoint_gradient} obtained via the adjoint state method is, as expected, equal to the correct gradient whose coordinates are given by Eq.~\eqref{eq:gradient_usual}.

\subsubsection{Proofs of the JVP and VJP expressions}
\label{sec:proof_jvp_vjp}
\begin{proof}[Proof of Lemma~\ref{lem:jvp} (JVP)]\label{proof:jvp}
    We denote by $\delta U(t)$ the first order variation of the trajectory $U(t)$ solution to~\eqref{eq:cauchy} in response to a perturbation $\delta \z$ of $\z$
    \begin{equation}\delta U(t) = \left.\frac{\mathrm{d}}{\mathrm{d}\varepsilon} U_{\z + \varepsilon \delta \z}(t)\right|_{\varepsilon = 0}\end{equation}
    By definition $E'(\z)\delta \z = \delta U(1)$.
    In Equation~\eqref{eq:cauchy}, keep in mind that $U$ itself depends on $\z$. Then differentiating with respect to $\z$ the initial condition and both sides of the differential equation yields that the first variation of $U$
    is the solution starting at $\delta U(0) = 0$ to $\dot{\delta U} = -\iu U(A + B)$ where 
    \begin{align}
        A &= \sum_{j=1}^{m}\RE \tr\left((\delta U)^\dagger H_j U \z + U^\dagger H_j \delta U \z + h_j\delta \z\right)h_j, \\ 
        B &= U^\dagger H \delta U.
    \end{align}
    We can simplify the $A$ term by observing that $\RE \tr ((\delta U)^\dagger H_j U\z) = \RE \tr (U^\dagger H_j \delta U\z)$, which means that the two first contributions from the left are equal. Then multiplying $\delta U$ by the inverse propagator $U^\dagger$ will cancel the contribution of the $B$ term. Indeed if $w(t) = \iu U(t)^\dagger \delta U(t)$ then $\dot{w} = \iu \dot{U}^\dagger \delta U + (A + B)$ and, given the differential equation on $\x$, the first term simplifies to $\iu \dot{\x}^\dagger \delta \x = -B$. Altogether we see that $\dot{w} = A$.
    Substituting $\delta U = -\iu U w$ in the $A$ terms yields the Equation~\eqref{eq:linearization} on $w$ because $w(0) =0$ since $\delta U(0) = 0$.
\end{proof}

\begin{proof}[Proof of Lemma~\ref{lem:vjp} (VJP)]\label{proof:vjp}
    The aim is to find the Hermitian traceless matrix $R$ such that 
    \begin{equation}\label{eq:def_adjoint}
        \RE \tr \left(S^\dagger \delta U(1)\right) = \RE \tr \left(R\delta \z\right)
    \end{equation}
    for any perturbation $\delta \z$ of $\z$.
    Define $w = \iu U \delta U$ as in Lemma~\ref{lem:jvp}. Observe that for any function $v$
    \begin{align}\label{eq:ftc}
        \RE \tr \left(v(1)^\dagger w(1)\right) = \int_{0}^{1}\RE \tr \left(\dot{v}^\dagger w + v^\dagger \dot{w}\right)\d t
    \end{align}
    by the fundamental theorem of calculus. Using the differential equation on $w$ lets us rewrite Equation~\eqref{eq:ftc} as
    \begin{align}
        \RE \tr \left(v(1)^\dagger w(1)\right) &- \int_{0}^{1}\RE \tr \left( (\dot{v} + 2\iu \sum_{j=1}^{m}\RE \tr(v^\dagger h_j)h_j\z)^\dagger w\right) \d t\\ = & \int_0^1 \RE \tr \left(\sum_{j=1}^{m}\RE \tr(v^\dagger h_j)h_j \delta \z \right)\d t 
    \end{align}
    Notice the right-hand side only involves the perturbations $\delta \z$ but not $w$. This means that choosing $v$ as the solution to Equation~\eqref{eq:adjoint_state}, the left-hand side becomes equal to the left-hand side of~\eqref{eq:def_adjoint}. Then choosing $R$ as in Eq.~\eqref{eq:grad} gives~\eqref{eq:adjoint_state} and concludes the proof.
\end{proof}

\subsection{Optimization routine: natural gradient descent} \label{sub:natural_gradient}
Consider Problem~\eqref{ocp:z}.
Instead of optimizing $|P(\z)|$ subject to the constraint $E(\z) = \target$, we simplify the problem by optimizing a functional $J(\z) = \ell(E(\z))$ penalizing the discrepancy to the target $\target$. In other words, we look for a $\z$ that is admissible for Problem~\eqref{ocp:z} but need not be time-optimal. However, this is done using a so-called natural gradient descent~(see e.g.~\cite{nurbekyan2023efficient}), a variant of the Gauss-Newton method for optimization. Its update rule $\z \mapsto \hat{\z}$ reads 
\begin{equation}\label{eq:natgrad}
    \hat{\z} = \z  -\rho \left[E'\left(\z\right)^\dagger E'\left(\z\right) + \alpha I\right]^{-1}\nabla J\left(\z\right), \quad \text{where} \quad \nabla J\left(\z\right) = E'(\z)^\dagger \nabla \ell(E(\z)).
\end{equation} 
In Section~\ref{sec:natgrad_details} below we provide different ways to gain intuition on this update rule.
Via the Tikhonov regularization parameter $\alpha$, this method allows for a natural way to penalize the variation of the Frobenius (Hilbert-Schmidt) norm of $\z$ and hence of the time $T = |P(\z)|$. 
Put differently, this means that we do not optimize the time directly but rather attempt to find an admissible point $\z$ for which $|P(\z)|$ is just as large as it needs to be for reaching the target.

We use an optimize-then-discretize approach, also known as \emph{indirect} method. This means that we first compute a descent direction such as $-\nabla J(\z)$ and only then introduce a numerical scheme to discretize the equations of motion. This is contrast to the use of an automatic differentiation engine which is part of the discretize-then-optimize framework, also known as \emph{direct} method.

A proof of concept written in Julia is available on the GitHub repository \href{https://github.com/killianlutz/Magicarp}{\texttt{Magicarp}}. Alternatively a Python code relying on JAX and the direct method is available on the repository \href{https://github.com/killianlutz/pyMagicarp}{\texttt{pyMagicarp}}.

Using a standard gradient descent with line search, infidelity plateaus are observed when the qudit dimension becomes large, say $d \geq 8$. This means that regardless of the number of gradient iterations, the value of $J$ remains above a "too large" positive value even though the minimal value of $J$ is known to be zero.
We resort to natural gradient descent to overcome this issue. This method has proven to be more efficient in practice. This can be understood since it incorporates second order information on the cost functional $J$.

More concretely, the hessian $J''(\xi)$ of the composition $J(\xi) = \ell(E(\xi))$ reads $J''(\xi) = S + R$ where $S = E'(\xi)^\dagger\ell''(E(\xi))E'(\xi)$ and $R$ involves $E''(\xi)$ and $\nabla \ell(\xi)$. For convex $\ell$, the matrix $S$ is positive semi-definite while the matrix $R$ need not have a "sign". In a way, $R$ is not as well behaved as $S$ for optimization purposes. In a Gauss-Newton method, $R$ is discarded and the matrix $S$ is used. In a (non-regularized) natural gradient method, $R$ is discarded and the matrix $S$ is used as if the hessian of $\ell$ was constant equal to the identity matrix.

\subsubsection{Principle of natural gradient descent}
Up to choosing an orthonormal basis $(\z_i)_{1 \leq i \leq n}$ of $\iu \mathfrak{su}(d)$, we may identify $\z$ with its coordinates $\xi_i = \langle \z, \z_i \rangle$. Accordingly we may view $E$ as a map from $\RR^n \to \CC^{d \times d}$.

Let $J : \RR^n \to \RR$ be the objective function. Assume that it splits into the composition of $E$ with a convex function $\ell : \CC^{d \times d} \to \RR$, so that $J = \ell \circ E$.
In practice $\ell(\x) = |\x - \target|^2$ or the infidelity $\ell(\x) = 1 - |\operatorname{Tr}(\x\target^\dagger)|/d$. Although the latter is a natural choice, it is \emph{not} convex. 

\paragraph{Main idea.}\label{sec:natgrad_details} The intuition behind splitting $J$ as $\ell$ composed with the model or end-point mapping $E$ is the following: when $\ell$ is convex, optimizing in the $\x$ space is in general far easier than optimizing $J$ directly in the $\z$ space due to non-convexity. 
To exploit this idea, the natural gradient step takes into account the model parametrization $\z \mapsto \x(1) = E(\z)$. It projects $-\nabla \ell(\x(1))$, the "good" descent direction in $\x$ space, onto the range of the linearization of $E$ at $\z$. Roughly speaking, the range of $E'(\z)$ describes the directions of motion locally available around $\x(1) = E(\z)$ within the state set $S = \lbrace E(\z), \, \z \in \iu \mathfrak{su}(d) \rbrace$. This projection can be seen as preconditioning $\nabla J(\z)$, the "bad" gradient in $\z$ space, by the regularized Gram matrix $E'(\z)^\dagger E'(\z)$ which describes how the parameters or degrees of freedom $\z$ offer an "independent" description of the state set $S$ around the state $E(\z)$.

\paragraph{Details.}
Fix a penalty $\alpha > 0$.
Given a point $\xi$ with associated trajectory $\x$, we update $\xi$ by a step in the direction $\mu$ with step-size $\rho$
\begin{equation}
    \delta \xi = \rho \frac{\mu}{\hnorm{\mu}}.
\end{equation}
The descent direction $\mu$ is determined as the solution to the Tikhonov-regularized least-squares problem
\begin{equation}
    \operatorname{\inf}_{\mu \in \RR^n} \hnorm{E'(\xi)\mu + \nabla \ell(E(\xi))}^2 + \alpha \hnorm{\mu}^2
\end{equation}
for which the normal equations are
\begin{equation}
    \left[E'(\xi)^\dagger E'(\xi) + \alpha I \right]\mu = -\nabla J(\xi).
\end{equation}
because $\nabla J(\xi) = E'(\xi)^\dagger \nabla \ell(E(\xi))$.
Given the solution $\mu$,the step-size $\rho$ minimizes on $(0, +\infty)$ the map
\begin{equation}
    \rho \mapsto J\left(\xi + \rho \frac{\mu}{\hnorm{\mu}}\right)
\end{equation}

\begin{remark}
    If $\xi$ is not a critical point of $J$, that is $\nabla J(\xi) \neq 0$ then any step $\mu$ defined by a symmetric semi-positive definite system of linear equation $A \mu = -\nabla J(\xi)$ is a descent direction for $J$. That is for such a $\mu$, $J(\xi + \varepsilon \mu) - J(\xi) = -\varepsilon \langle A \mu, \mu \rangle + o(\varepsilon)$ is negative for small enough $\varepsilon$. In particular this holds for any $\alpha \geq 0$ when $A = E'(\xi)^\dagger E'(\xi) + \alpha I$.
\end{remark}

\subsubsection{Implementation details.}
\paragraph{Algorithm to calculate JVP: $E'(\z)\delta \z$.}
Given $\z$, we solve forward from $t = 0$ to $t = 1$ the augmented system $(\x, \delta \x)$ starting at $(I, 0)$ where $\delta \x(t) = -\iu \x(t)w(t)$ is the linearization of the trajectory $\x$ in the direction $\delta \z$ and $w$ is defined by Equation~\eqref{eq:linearization}. The value at time $t = 1$ of $\delta \x$ is the desired quantity $E'(\z)\delta \z$.

\paragraph{Algorithm to calculate VJP: $E'(\z)^\dagger S$.}
Given $\z$, we solve for the terminal state $\x(1)$ and evolve \emph{backward} from $t=1$ to $t=0$ the augmented system $x = (\x, v, z)$ starting at $x_1(1) = \x(1)$, $x_2(1) = \iu x_1^\dagger S$ and $x_3(1) = 0$. The equation of motion for $x_1$ and $x_2 = v$ are respectively given by~\eqref{eq:cauchy} and \eqref{eq:adjoint_state} with minus signs to account for backward motion, and that of $x_3 = z$ is $\dot{z} = -\sum_{j=1}^{m}\RE \tr(h_jv)h_j$. In this way the value at time $t = 0$ of $z$ is the desired quantity $R = E'(\z)^\dagger S$.

\paragraph{Algorithm to calculate $\nabla J(\z)$:} Given $\z$, apply (ii) with $S = \nabla \ell(E(\z))$.

\paragraph{Algorithm to calculate combined VJP and JVP: $E'(\z)^\dagger E'(z)\delta \z$.}
Find the update step of the natural gradient descent requires solving a symmetric positive-definite linear system involving the gram matrix of $E'(\z)$ regularized by $\alpha I$. Since $\alpha > 0$ and to avoid creating the large non-sparse matrix $E'(\z)$ of order scaling like $d^2$, we appeal to matrix-free methods of Krylov-type. For instance GMRES~\cite{saad1986gmres,zou2023gmres}.

Given $\delta \z$, we apply (a) to get $S := E'(\z)\delta \z$ and then (b) to obtain $R = E'(\z)^\dagger S$. Then $R + \alpha \delta \z$ is the desired image of $\delta \z$ under the linear map defining the normal equations. The step-size $\rho$ is calculated using the golden-section method, a gradient-free optimization method for unimodal functions~\cite{press2007numerical}. This assumption is usually not verified but the method empirically works decently well in our setting.

\paragraph{Initialization.} 
Whenever $\z^{(0)}$ is orthogonal to the linear span of the control Hamiltonians, it is an equilibrium point of the optimization procedure. In particular for some targets $\target$, the reasonable initial guess $\z^{(0)} = \iu \log(\target)$ could be one such equilibrium point. 

Accordingly a random initial-guess is preferable without further a priori knowledge.
For the purpose of minimizing $T = |P(\z)|$ subject to $E(\z) = \target$, it is additionally preferable for $|P(\z^{(0)})|$ to be relatively "small" and to choose the penalty $\alpha$ quite large. 

Typically $\alpha = 10^{-3}$ is a robust choice with respect to the choice of target $\target$ and dimension $d$. If $\alpha$ is larger, we expect $|P(\z)|$ to be smaller when the algorithm stops, but this may require a lot more additional iterations.

\paragraph{Validation.} Our stopping criterion is based on the value of the objective function and of the form $J(\z) \leq \tau$ where $\epsilon$ is a given absolute tolerance. Once the algorithm stops, we validate the candidate numerical solution $\z$ using an independent and high-fidelity ODE solver to compute the $\x(1)$ associated to this $\z$. We evaluate the corresponding discrepancy to the target $J_{\rm val}(\z)$ and check whether it is indeed less then or equal to $\tau$. Typically $\tau = 10^{-4}$ if $\ell$ is the infidelity $\ell(x) = 1 - |\operatorname{Tr}(\x\target^\dagger)|/d$.

\paragraph{(vii) Mesh-refinement.} Given a uniform grid of the interval $[0, 1]$ with stepsize $\delta t$, we optimize until the stopping criterion $J(\z) \leq \tau$ is met. If the result is not validated in the sense of (f), that is $J_{\rm val}(\z) > \tau$, then we refine the mesh. For instance by a factor two $\delta t_{\rm new} = \delta t/2$. Keeping $\z$ as initial guess, we then repeat the previous steps until $J_{\rm val} \leq \tau$ or the number of iterations exceeds a pre-defined threshold. This can be seen as a continuation method over the mesh-size.

\section{Population dynamics for a QFT on the triple decker}~\label{app:triple_decker_dynamics}

\begin{figure}[!htbp]
    \centering
    \includegraphics[width=\linewidth]{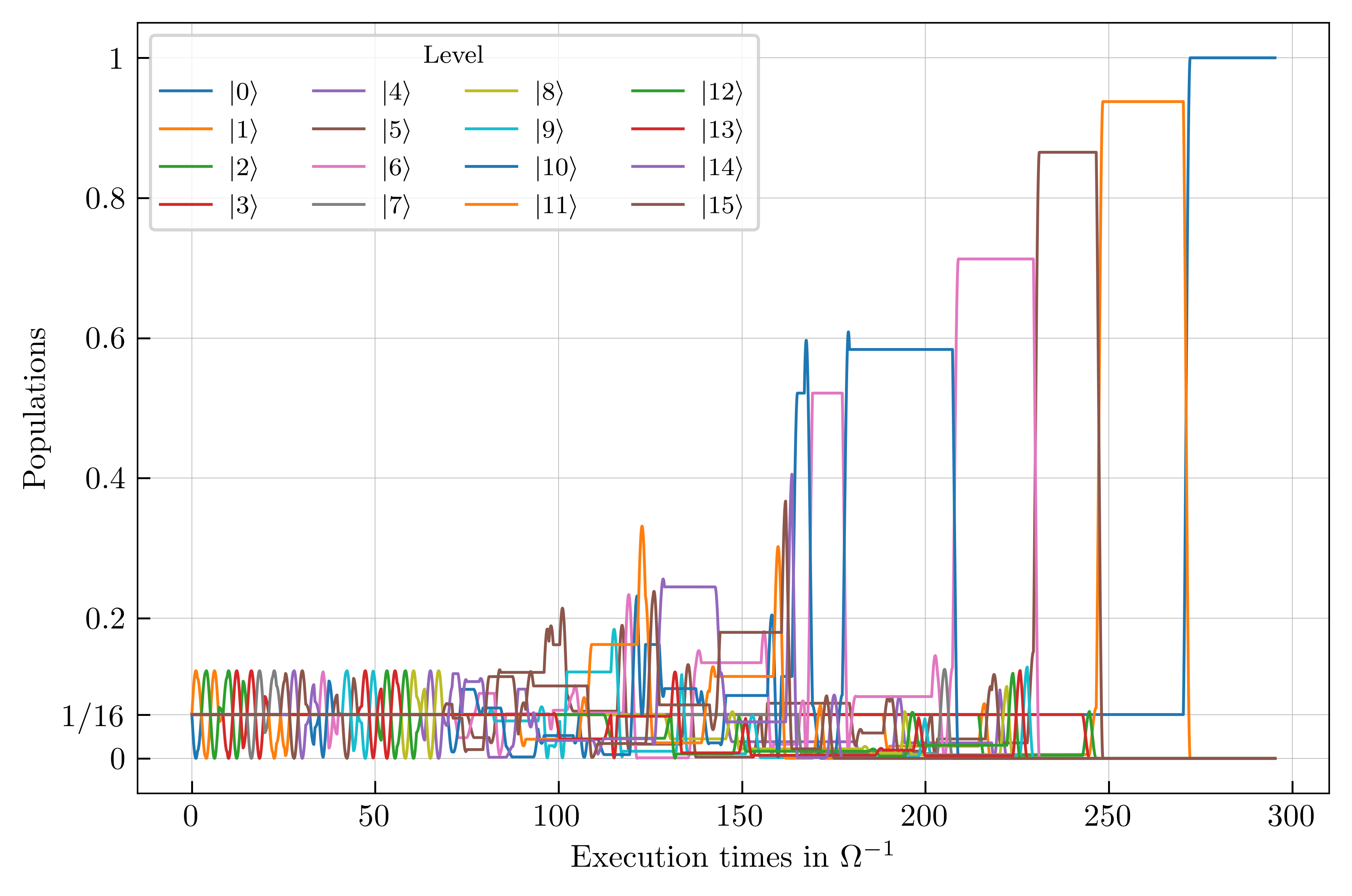}
    \caption{Time evolution of the energy levels for the triple decker subject to the control pulses in~\cref{fig:controls} which were calculated using the GRD to implement the QFT. The initial state is a superposition of all the basis states with equal weights and so it ends-up in the state $\ket{0}$.}\label{fig:pop_dynamics}
\end{figure}
\begin{figure}[!htbp]
    \centering
    \includegraphics[height=0.9\textheight]{imgs/Grover_pulses/controls_with_thumbnails.png}
    \caption{Time evolution of the controls calculated using the GRD to implement the QFT on the triple decker.\ Each control correspond to an addressable energy difference which are indicated by edges in the coupling graph of~\cref{fig:graph}. The dynamics of the populations are shown in~\cref{fig:pop_dynamics}. Notice that some controls are not used. This reflects the non-unique arbitrary path taken to reach each nodes, similar to the Travelling Salesman problem.}\label{fig:controls}
\end{figure}

\clearpage
\newpage
\bibliography{refs,biblio_math}
\end{document}